\documentclass{article}
\usepackage{kr}
\usepackage{times}
\usepackage{soul}
\usepackage{url}
\usepackage[hidelinks]{hyperref}
\usepackage[utf8]{inputenc}
\usepackage[small]{caption}
\usepackage{graphicx}
\usepackage{amsmath}
\usepackage{amsthm}
\usepackage{booktabs}
\usepackage[linesnumbered,ruled,vlined]{algorithm2e}
\usepackage{algpseudocode}
\usepackage{todonotes}
\usepackage{cleveref}
\usepackage{comment}
\usepackage{thmtools}
\usepackage{thm-restate}
\newtheorem{definition}{Definition}

\newtheorem{theorem}{Theorem}

\newtheorem{lemma}{Lemma}
\newtheorem{corollary}{Corollary}
\newtheorem{remark}{Remark}

\usepackage{amssymb}
\usepackage{xspace}
\usepackage{paralist}
\usepackage{mathtools}

\usepackage{todonotes}
\usepackage{subcaption}

\usepackage{tikz}
 \usetikzlibrary{trees}
 \usetikzlibrary{shapes}
 \usetikzlibrary{fit}
 \usetikzlibrary{shadows}
 \usetikzlibrary{backgrounds}
 \usetikzlibrary{arrows,automata}
 \usetikzlibrary{positioning,fit,arrows.meta,backgrounds}
  \usetikzlibrary{calc, patterns, arrows, cd, tikzmark, decorations.pathreplacing, calligraphy, matrix, decorations.pathmorphing, shapes.multipart, plotmarks}

\tikzset{
    module/.style={%
        draw, rounded corners,
        minimum width=#1,
        minimum height=7mm,
        font=\sffamily
        },
    module/.default=2cm,
    >=LaTeX
}

\tikzset{
   n/.style= {circle,fill,inner sep=1.5pt,node distance=2cm}
  ,acc/.style={circle,draw,inner sep=3pt,node distance=2cm}
  ,phantom/.style={circle},
  ,arr/.style={->, >=stealth, semithick, shorten <= 3pt, shorten >= 3pt}
}

\usepackage{pgfplots}

\newcommand{\LTLf}{LTL$_f$\xspace}
\newcommand{\PPLTL}{PPLTL\xspace}
\newcommand{\LTL}{LTL\xspace}

\newcommand{\LTLfp}{LTL$_f$+\xspace}
\newcommand{\PPLTLp}{PPLTL+\xspace}

\newcommand{\DFA}{{DFA}\xspace}
\newcommand{\NFA}{{NFA}\xspace}

\newcommand{\tool}{$\mathsf{LydiaSyft+}$\xspace}

\newcommand{\rank}{\mathsf{rank}}

\title{Symbolic Synthesis for \LTLfp Obligations}

\author{%
Giuseppe {De Giacomo}$^{1}$\and
Christian Hagemeier$^1$\and
Daniel Hausmann$^2$\and
Nir Piterman$^{3}$
\affiliations
$^1$University of Oxford, UK\\
$^2$University of Liverpool, UK\\
$^3$University of Gothenburg and Chalmers University of Technology, Sweden\\
\emails
hausmann@liverpool.ac.uk,
\{christian.hagemeier,giuseppe.degiacomo\}@cs.ox.ac.uk,
piterman@chalmers.se
}

\begin{document}
\maketitle

\begin{abstract}
    We study synthesis for obligation properties expressed in \LTLfp, the extension of \LTLf to infinite traces. 
    Obligation properties are positive Boolean combinations of safety and guarantee (co-safety) properties and form the second level of the temporal hierarchy of Manna and Pnueli.
    Although obligation properties are expressed over infinite traces, they retain most of the simplicity of \LTLf. 
    In particular, we show that they admit a translation into symbolically represented deterministic weak automata (DWA) obtained directly from the symbolic deterministic finite automata (DFA) for the underlying \LTLf properties on trace prefixes. 
    DWA inherit many of the attractive algorithmic features of DFA, including Boolean closure and polynomial-time minimization. 
    Moreover, we show that synthesis for \LTLfp obligation properties is theoretically highly efficient~--~solvable in linear time once the DWA is constructed.
    We investigate several symbolic algorithms for solving DWA games that arise in the synthesis of obligation properties and evaluate their effectiveness experimentally.
    Overall, the results indicate that synthesis for \LTLfp obligation properties can be performed with virtually the same effectiveness as \LTLf synthesis.
    \end{abstract}

\section{Introduction}\label{sec:intro}
Linear Temporal Logic (LTL)~\cite{Pnueli77} and its finite-trace variants are commonly used in computer science and artificial intelligence, e.g., in planning for temporally extended goals and declarative control knowledge~\cite{DBLP:journals/amai/BacchusK98,DeGiacomoV99,BacchusK00,CalvaneseGV02,BaierM06,BaierFM07,GereviniHLSD09}.

Reactive synthesis concerns the automatic construction of programs (typically called \emph{strategies}) from temporal specifications for systems
(e.g. agents, processes, protocols, controllers, or robots) that interact with their environments during execution \cite{PnueliR89,finkbeiner2016synthesis,EhlersLTV17}. 
It is closely related to strong planning for temporally extended goals in fully observable nondeterministic domains \cite{Cimatti03,DBLP:journals/amai/BacchusK98,BacchusK00,CalvaneseGV02,BaierFM07,GereviniHLSD09,DR-IJCAI18,CamachoBM19}.
Significant advances in reactive synthesis have been achieved using the GR(1) fragment of \LTL \cite{PiPS06}, and by restricting attention to finite traces using \LTLf \cite{DegVa15}. These successes were largely based on the use of symbolic
techniques for handling temporal formulas, which substantially enhances scalability and computational capacity.

Recently, the symbolic techniques underlying \LTLf have been extended to infinite traces through the logic \LTLfp \cite{LTLfplus}.
The key promise of \LTLfp is that finite automata (as used in \LTLf) and the symbolic representations that efficiently support
them can also be leveraged in the setting of infinite traces.
\LTLfp builds on the so-called \emph{Manna-Pnueli} (or \emph{safety-progress}) \emph{hierarchy} of temporal formulas.
This hierarchy was originally introduced in \cite{LichtensteinPZ85} and subsequently developed in detail by Manna and Pnueli in \cite{DBLP:conf/podc/MannaP89} and \cite{MannaPnueli92,MannaPnueli95,MannaPnueli10}; see also the survey~\cite{PitermanP18}. 
In this paper, we concentrate on the lowest levels of the hierarchy,
namely \emph{safety} and \emph{guarantee} (also known as \emph{co-safety}) properties, as well as their Boolean closure, referred to as \emph{obligation} properties, which are our main focus. 

Obligation properties are very common in 
practice.  Specifications used in
model checking or synthesis frequently belong to the lower levels of
the temporal hierarchy, namely safety, guarantee, and obligation.
For instance, among the 55 specification patterns
of Dwyer et al.~\cite{dwyer.98.fmsp}, 25 are obligation properties.
Similarly, Somenzi \& Bloem's compilation of 25 LTL formulas ``found
in the literature''~\cite{somenzi.00.cav} contains 13
obligation properties.  
Moreover, even in the \LTLf{} setting, assumptions about the environments are often forms of obligations \cite{AminofGSFRZ25}. 

For obligation properties, 
we can work with deterministic weak automata (DWA)~\cite{DBLP:journals/ipl/Loding01}. 
These automata have a substantially simpler structure than those required for arbitrary temporal formulas, particularly
formulas from the higher levels of the Manna-Pnueli hierarchy.

Their structural simplicity makes DWA very easy to use: they
are closed under Boolean operations using constructions analogous to those
of deterministic finite automata (DFA), and after a linear-time preprocessing
stage~\cite{DBLP:journals/ipl/Loding01}, they can be minimized using the same algorithm as DFA.

In this setting, the contributions of the present paper are as follows:
\begin{enumerate}
    \item A technique for constructing symbolically represented DWA from \LTLfp{} obligation properties.
    \item A reduction of synthesis for the obligation fragment of \LTLfp{} to the solution of weak games played on DWA.
    \item A novel symbolic algorithm for the solving weak games by alternating safety and reachability computations. Together with three existing solution algorithms (the standard algorithms for B\"uchi and co\mbox{-}B\"uchi games, and a linear-time algorithm based on SCC decomposition), we obtain four fully symbolic \textsc{2ExpTime} synthesis algorithms for \LTLfp{} obligation formulas.
    \item A proof-of-concept implementation and experimental evaluation of these algorithms,  based on \tool~\cite{10.24963/kr.2025/78}, demonstrating that synthesis for the \LTLfp{} obligation fragment can be solved very efficiently, achieving performance comparable to synthesis for \LTLf.
\end{enumerate}

\section{Preliminaries}\label{sec:prelim}

\paragraph{\LTL and \LTLf.}\label{par:logics}
\emph{Linear-time Temporal Logic} (\LTL) specifies temporal properties of infinite traces~\cite{Pnueli77}.
We use the standard syntax
\begin{align*}
\varphi,\psi::= p \mid \neg\varphi\mid \varphi\land\psi\mid \mathsf{X}\, \varphi \mid \varphi\,\mathsf{U}\,\psi \tag{$p\in AP$},
\end{align*}
where $AP$ is a countable set of atomic propositions.
Common abbreviations include $\varphi\lor\psi=\neg(\neg \varphi\land\neg\psi)$, 
$\mathsf{true}= p\lor \neg p$, $\mathsf{false}= p\land \neg p$, 
$\mathsf{F}\,\varphi=\mathsf{true}\,\mathsf{U}\,\varphi$ (``eventually'') and
$\mathsf{G}\varphi=\neg \mathsf{F}\neg\varphi$ (``always'').
Formulas are evaluated over infinite traces $\tau\in (2^{AP})^\omega$.
\LTL is expressively equivalent to first-order logic over infinite traces~\cite{Kamp1968-KAMTLA-4}.

\emph{\LTL on finite traces} (\LTLf)~\cite{DegVa13} uses the same syntax 
but is interpreted over finite traces $\tau\in (2^{AP})^*$.
In this setting, $\mathsf{X}\,\mathsf{false}$ holds exactly at the last position of a finite trace.
A common abbreviation is the operator $\mathsf{X}[!]\varphi=\neg \mathsf{X}\neg\varphi$ (``strong next''),
expressing the existence of a next position in the trace that satisfies $\varphi$.
Given an \LTLf formula $\Phi$, let 
$[\Phi]=\{\tau\in (2^{AP})^*\mid \tau,0\models\Phi\}$ denote
the set of finite traces that satisfy $\Phi$ at position~$0$.
\LTLf is expressively equivalent to first-order logic over finite traces and to star-free regular expressions, see~\cite{DegVa13}.

\paragraph{Automata on finite traces.}\label{par:automata}
A transition system $T=(\Sigma,Q,I,\delta)$ consists
of a finite alphabet $\Sigma$, a finite set $Q$ of states, a set $I\subseteq Q$
of initial states, and a transition relation $\delta\subseteq Q\times \Sigma\times Q$.
For $q\in Q$ and $a\in\Sigma$, define $\delta(q,a)=\{q'\in Q\mid (q,a,q')\in \delta\}$.
The system is \emph{deterministic} if $|I|=1$ and $|\delta(q,a)|=1$ for all
$q\in Q$ and $a\in\Sigma$, and
\emph{nondeterministic} otherwise; in the deterministic case, we write $\delta(q,a)=q'$ for
the unique successor and denote the initial state by $\iota$.
A \emph{finite automaton} $\mathcal{A}=(T,F)$ is a transition system together with a set $F\subseteq Q$ of accepting states.
If $T$ is deterministic, $\mathcal{A}$ is a deterministic finite automaton (\DFA); otherwise it is a nondeterministic finite automaton (\NFA).
A \emph{run} of $\mathcal{A}$ on a word $w\in \Sigma^*$ is a path starting from an initial state whose sequence of transition labels  is $w$;
it is accepting if it ends in a state from $F$.
The automaton accepts the language $L(\mathcal{A})$ consisting of all finite words for which
an accepting run exists.

For each \LTLf formula $\Phi$, one can construct 
an equivalent \NFA of size $2^{\mathcal{O}(|\Phi|)}$ and an equivalent \DFA of size $2^{2^{\mathcal{O}(|\Phi|)}}$ represented symbolically \cite{DegVa15}.

\smallskip
\paragraph{Infinite-duration games on finite graphs.}\label{par:games}
A \emph{game arena} is a finite directed graph $A=(V,E\subseteq V\times V)$ whose
vertex set is partitioned into system nodes $V_s$ and environment nodes $V_e$.  
For $v\in V$, define $E(v)=\{v'\mid (v,v')\in E\}$ and assume $E(v)\neq\emptyset$ for all $v$. 
A \emph{play} is a path in $A$. 
A \emph{memoryless strategy} for the system player is a function $\sigma:V_s \to V$ that assigns 
to every system node $v\in V_s$ 
a successor $\sigma(v)\in E(v)$. A play $v_1 v_2 \ldots$
is \emph{compatible} with a $\sigma$ if whenever $v_i\in V_s$, we have $v_{i+1}=\sigma(v_i)$.
An \emph{objective} is a set of infinite plays; a play is \emph{winning}
for the system player if belongs to the objective.
We consider \emph{B\"uchi} and \emph{co\mbox{-}B\"uchi} objectives specified
by sets $F\subseteq V$: B\"uchi objectives require visiting $F$ infinitely often,
whereas co\mbox{-}B\"uchi objective require visiting $V\setminus F$ finitely often.
Strategies and winning plays for the environment player are defined dually.
A strategy for a player wins a node $v$ 
if every compatible play starting at $v$ is winning for that player.
B\"uchi and co\mbox{-}B\"uchi games are positionally determined:
every node is won by exactly one player, and memoryless (positional) strategies suffice.
Solving a game amounts to computing winning regions together
with witnessing strategies.

\smallskip
\paragraph{Reactive synthesis and games.}\label{par:syntgames}
Assume that the set $AP$ of atomic propositions is
partitioned into system actions $X$ and environment actions $Y$.
A \emph{(synthesis) strategy} is a function $\sigma:(2^Y)^*\to 2^X$. An \emph{outcome}
of $\sigma$ is an infinite word $(x_1\cup y_1)(x_2\cup y_2)\ldots\in (2^{AP})^\omega$
such that $x_{i+1}=\sigma(y_1 y_2\ldots y_i)$ for all $i\geq 0$.
Thus, synthesis strategies encode transducers.
The \emph{synthesis problem} for a temporal formula $\Psi$ 
asks whether there exists a synthesis strategy $\sigma$ such that every 
outcome of $\sigma$ satisfies $\Psi$, and if so, to construct one.

A deterministic transition system $T=(2^{\mathsf{AP}},Q,\iota,\delta)$ induces
a game arena $A_T=(Q\cup Q\times 2^X\cup Q\times 2^X\times 2^Y,E)$ in which the system owns nodes in $Q$
and the environment owns all other nodes. Transitions are defined by $E(q)=\{q\}\times 2^X$,
$E(q,x)=\{(q,x)\}\times 2^Y$, and  $E(q,x,y)=\{\delta(q,x\cup y)\}$.
Plays of the form $q_1(q_1,x_1)(q_1,x_1,y_1)q_2(q_2,x_2)(q_2,x_2,y_2)\ldots$ induce runs
$q_1 q_2\ldots$ of $T$ on the word $(x_1\cup y_1)(x_2\cup y_2)\ldots$.

Hence, the synthesis problem reduces to translating the specification into a deterministic
transition system and solving the induced game with an appropriate objective.

\section{Specification of Obligation Properties}

In their seminal work, Manna and Pnueli introduced a hierarchy of $\omega$-regular languages~\cite{DBLP:conf/podc/MannaP89}.
We apply this hierarchy to LTL-definable properties, i.e., properties definable
in first-order logic over infinite sequences. 
The hierarchy is based on four fundamental classes of infinite-trace properties
obtained from finite-trace properties via trace quantification.
\emph{Safety} properties require that a finite-trace property holds
for all prefixes, whereas \emph{guarantee} properties require the existence of a prefix satisfying 
the finite-trace property.
\emph{Recurrence} properties require that the finite-trace property holds on infinitely many prefixes,
and \emph{persistence} properties require it to hold on all but finitely many prefixes.
\emph{Obligation} properties are Boolean combinations of safety and guarantee properties,
and \emph{reactivity} properties are Boolean combinations 
of recurrence and persistence properties.

In this work, we focus on obligation properties.
These properties naturally capture specifications that combine
invariants with eventual goals and therefore arise frequently in planning, 
control, and reactive synthesis.

The recently proposed logic \LTLfp~\cite{LTLfplus}, which is inspired by the Manna-Pnueli hierarchy, uses \LTLf properties to define the finite-trace properties mentioned in the hierarchy description.
Namely, it explicitly uses guarantee, safety, recurrence, and persistence properties over \LTLf properties, and considers Boolean combinations of them. 

Formulas of \LTLfp over a countable set $AP$ of propositions are generated by the grammar
\begin{align*}
\Psi, \Psi' ::= \forall\Phi \mid \exists\Phi \mid \forall\exists\Phi \mid \exists\forall\Phi \mid \Psi \lor \Psi' \mid \Psi \land \Psi' \mid \neg \Psi
\end{align*}
where $\Phi$ is an \LTLf formula over $AP$.
We refer to formulas of the form $\mathbb{Q}\Phi$ as \emph{finite-trace components}.
Recall that $[\Phi]\subseteq (2^{AP})^*$ denotes the set of finite traces satisfying $\Phi$. 
Given a set $R\subseteq (2^{AP})^*$ of finite traces, let $\exists R$ ($\forall R$) denote the set of infinite traces $\tau\in (2^{AP})^\omega$
such that at least one prefix of $\tau$ (respectively, every prefix) belongs to $T$.
Similarly, let $\forall\exists T$ ($\exists\forall T$) denote the set of infinite traces for which infinitely many 
prefixes (respectively, all but finitely many prefixes) belong to $T$.
We evaluate \LTLfp formulas $\Psi$ over infinite traces using the extension $[\Psi]\subseteq (2^{AP})^\omega$
defined inductively by
$[\Psi\lor\Psi'] = [\Psi]\cup [\Psi']$, 
$[\Psi\land\Psi'] = [\Psi]\cap [\Psi']$,
$[\neg\Psi] = (2^{AP})^\omega\setminus[\Psi]$, and $[\mathbb{Q}\Phi] = \mathbb{Q} [\Phi]$
where $\mathbb{Q}\in\{\exists,\forall,\forall\exists,\exists\forall\}$.

\begin{theorem}\cite{LTLfplus}
The logics \LTLfp and \LTL define the same infinite-trace properties, namely those definable in first-order logic over infinite sequences.
\end{theorem}

The \emph{obligation fragment} of \LTLfp is obtained by allowing only safety and guarantee finite-trace components,
i.e., by using only the clauses $\exists\Phi$ and $\forall\Phi$ from the syntax.
We refer to formulas in this fragment as \emph{obligation formulas}.

As an example, suppose a domain and a collection of goals
are specified by finite-trace \LTLf formulas $\Phi_d$ and $\Phi_{g_1}$,\ldots,$\Phi_{g_k}$.
The obligation formula $\forall\Phi_d\to(\bigwedge_{i\leq k}\exists\Phi_{g_i})$
states that 
every infinite trace that remains within the domain eventually satisfies each goal. 
Goals can be made conditional using $\forall\Phi_d\to(\bigwedge_{i\leq k}(\exists\Phi_{t_i}\to\exists\Phi_{g_i}))$
where $\Phi_{t_i}$ specifies a triggering property for the $i$th goal. This formula states
that whenever a trigger occurs along a trace that stays within the domain,
the corresponding goal is eventually achieved (either before or after the trigger occurs).
Not every property is expressible as an obligation formula.
For example, the recurrence formula $\forall\exists (\mathsf{F}(a\land \mathsf{X}\,\mathsf{false}))$ states
that infinitely many prefixes end with $a$, equivalently, that $a$ holds infinitely often.

\begin{lemma}
The obligation fragment of \LTLfp defines exactly the obligation properties over infinite traces.
\end{lemma}

\begin{remark}
Alongside \LTLfp,~\cite{LTLfplus} introduces the related and equally expressive logic \PPLTLp, in which finite-trace specifications are formulated in \emph{Pure Past LTL} (\PPLTL) rather than in \LTLf. We focus our technical developments on the obligation fragment of \LTLfp, but emphasize that the synthesis algorithms
we propose extend directly to the obligation fragment of \PPLTLp.
Since \PPLTL formulas can be translated into DFA of single-exponential size,
the upper bounds in the \PPLTLp counterparts of 
Corollary~\ref{cor:to_automata} and Theorem~\ref{thm:synth} are singly exponential.
\end{remark}

The obligation fragment of \emph{\LTL}, obtained by imposing syntactic restrictions
on classical \LTL formulas, has the same expressive power
and thus also characterizes the obligation properties over infinite traces~\cite{ChangMP92}.

\begin{definition}
The \emph{synthesis problem} for obligation \LTLfp asks, given 
an obligation formula $\Psi$, whether there exists a strategy $\sigma$ such that every 
outcome of $\sigma$ satisfies $\Psi$, and if so, to construct one.
\end{definition}

\begin{restatable}{lemma}{simp}
\label{lem:simp}
Let $\Phi$ and $\Phi'$ be \LTLf formulas. Then 
\begin{align*}
\forall \Phi \land \forall \Phi' & \equiv  \forall (\Phi\land \Phi') &
\exists \Phi \lor \exists \Phi' & \equiv  \exists (\Phi\lor \Phi') 
\end{align*}
Moreover, $\Phi$ is equi-realizable with the \LTLfp formula $\exists \Phi$.
\end{restatable}

Since obligation formulas form a fragment of \LTLfp, their synthesis problem can 
be solved using techniques for the full logic, based on games over deterministic
Emerson-Lei automata (DELA), whose objectives are Boolean combinations of B\"uchi
and co\mbox{-}B\"uchi conditions~\cite{LTLfplus,10.24963/kr.2025/78}. These methods translate
individual finite-trace components into B\"uchi or co\mbox{-}B\"uchi automata and obtain a
DELA by Boolean combination of the individual automata.
The resulting game is solved via a nested fixpoint computation
derived from the Zielonka tree of the acceptance condition~\cite{DBLP:conf/fossacs/HausmannLP24}; 
solving DELA games is known to be \textsc{PSpace}-complete.

We show that obligation formulas admit a reduction to substantially simpler
$\omega$-automata, enabling efficient minimization
and linear-time game solving.

\section{From Obligations to DWA}\label{sec:obltodwa}

In this Section, we transform obligation formulas into deterministic \emph{weak} automata (DWA) on 
infinite words. In every run of a DWA, there exists a position from which on only accepting 
states or only rejecting states are visited. Owing to this structural property, DWA are 
simpler than B\"uchi and co\mbox{-}B\"uchi automata, and, in many respects, similar to DFA.
In particular, DWA are closed under conjunction and disjunction and can be minimized efficiently using a variant of
Hopcroft's DFA minimization algorithm. 

Each finite-trace component $\mathbb{Q}\Phi$ of an obligation formula corresponds to a 
DWA: the finite-trace formula $\Phi$ can be transformed to a DFA $\mathcal{A}_\Phi$ that accepts exactly
the finite traces that satisfy $\Phi$. The automaton for $\mathbb{Q}\Phi$ then processes 
infinite traces $\tau$ and checks whether $\mathcal{A}_\Phi$ accepts some prefix of $\tau$ (if $\mathbb{Q}=\exists$),
or all prefixes of $\tau$ (if $\mathbb{Q}=\forall$). This is achieved by turning 
accepting (respectively, rejecting) states in $\mathcal{A}_\Phi$ into absorbing sinks.
An infinite run of the resulting automaton $\mathcal{A}_{\mathbb{Q}\Phi}$ is accepting iff it visits
an accepting state infinitely often, which in turn is the case iff from some position onward only accepting states are visited.

As DWA are closed under conjunction and disjunction, we can construct DWA for full obligation formulas,
which are Boolean combinations of safety and guarantee finite-trace components.

\paragraph{Weak B\"uchi automata.}

We consider standard automata on infinite traces. 
A \emph{B\"uchi automaton} $\mathcal{A}=(T,F)$ is a transition
system $T=(\Sigma,Q,\iota,\delta)$ together with a set $F$ of accepting states. 
A \emph{run} of $T$ on an infinite word $w=a_1 a_2\ldots\in \Sigma^\omega$
is an infinite sequence $\pi=q_1 q_2\ldots\in Q^\omega$ of states such that
$q_1=\iota$ and $q_{i+1}\in \delta(q_i,a_i)$ for all $i\geq 1$.
The automaton recognizes
the ($\omega$-regular) language $L(\mathcal{A})$ of all words $w\in\Sigma^\omega$
for which there exists a run that visits states in $F$ infinitely often.
Dually, a \emph{co\mbox{-}B\"uchi automaton} recognizes the language of all words
admitting a run that visits rejecting states only finitely often.

A \emph{strongly connected component (SCC)} of a transition system $T=(\Sigma,Q,\iota,\delta)$
is a maximal set $S\subseteq Q$ of states such that every state in $S$ is reachable from every
other state in $S$. The SCC decomposition of $Q$ can be computed in time $\mathcal{O}(|Q|+|\delta|)$.
Reachability induces a partial order
on SCCs; a \emph{bottom SCC} is a minimal element of this order. 
A state is \emph{recurrent} if it can reach itself via a non-empty path,
and \emph{transient} otherwise.

An automaton $(T,F)$ is \emph{weak} if every SCC $S$ is either entirely accepting ($S\subseteq F$) or entirely rejecting ($S\cap F=\emptyset$)~\cite{DBLP:journals/ipl/Loding01}.\footnote{This definition is different from, though equivalent to, the one in Section 9.1.5 of~\cite{Hofmann2025}.}

For weak automata, B\"uchi and co\mbox{-}B\"uchi acceptance coincide: every infinite run that visits an accepting state 
infinitely often must eventually remain within an accepting SCC. Hence we simply refer to deterministic 
weak B\"uchi automata as \emph{deterministic weak automata} (DWA).

\paragraph{Closure properties of DWA.}

To translate obligation formulas into $\omega$-automata, we require closure under intersection and union.

Given deterministic transition systems $T_1=(\Sigma,Q_1,\iota_1,\delta_1)$ and $T_2=(\Sigma,Q_2,\iota_2,\delta_2)$,
their product is the transition system 
$T_1\otimes T_2=(\Sigma,Q_1\times Q_2, (\iota_1,\iota_2), \delta_\otimes)$ where
$\delta_\otimes((q_1,q_2),a)=(\delta_1(q_1,a),\delta_2(q_2,a))$.

\begin{restatable}{lemma}{closure}
\label{lem:weak_closure}
The set of languages recognizable by DWA is closed under union, intersection,
and complement.
\end{restatable}

\begin{proof}{\emph{(Sketch)}} 
Complementation is obtained by swapping accepting and rejecting states. Union and intersection are realized via the product construction with accepting sets $F_1\times F_2$ (intersection) and $(Q_1\times F_2)\cup(F_1\times Q_2)$ (union). All constructions preserve weakness.
\end{proof}

\paragraph{DWA minimization.}

For every DWA $\mathcal{A}$, there exists a unique minimal equivalent DWA
that can be computed
in time $\mathcal{O}(n\log n)$, where $n$ is the number of states~\cite{DBLP:journals/ipl/Loding01}.
DFA minimization using the Hopcroft algorithm can be applied once the automaton is brought into a suitable normal form. 
This form can be obtained in linear time by marking transient states accepting or rejecting according to their \emph{rank}.
Intuitively, the rank of a state $q$ is the maximum number of alternations between accepting and rejecting recurrent states
along runs starting from $q$. States with different ranks are not equivalent, as higher ranks induce additional accepted or 
rejected words. Transient states are marked accepting iff their rank is even, 
allowing Hopcroft's algorithm to merge them with maximally ranked successors.

Formally, define $\rank(S)$ for SCCs $S$ inductively.
Bottom accepting and rejecting SCCs have rank $0$ and $1$, respectively.
For a non-bottom SCC $S$, let $l(S)$ be the maximum rank among its successors in the SCC decomposition.
If $S$ is transient, set $\rank(S)=l(S)$. 
If $S$ is recurrent, set $\rank(S)=l(S)+1$ when parity disagrees with acceptance,
and $\rank(S)=l(S)$ otherwise.
The rank of a state is the rank of its SCC and can be computed in time $\mathcal{O}(|\mathcal{A}|)$. 
Hence, recurrent states have even rank iff they are accepting.

Given a DWA $\mathcal{A}=(T,F)$ with state set $Q$, define $F'=\{q\in Q\mid \rank(q)\text{ is even}\}$
and put $\mathcal{A'}=(T,F')$; this does not change the marking of recurrent states. Applying Hopcroft minimization to $\mathcal{A'}$ yields $\mathcal{A}_{\min}$.

\begin{lemma}\cite{DBLP:journals/ipl/Loding01}\label{lem:DWAmincorrectness}
$\mathcal{A_{\min}}$ is the minimal DWA equivalent to $\mathcal{A}$.
\end{lemma}

\begin{corollary}\label{cor:DWAmin}
Deterministic weak automata with $n$ states can be minimized in time $\mathcal{O}(n\log n)$.    
\end{corollary}

\paragraph{From obligation formulas to DWA.}

Recall that obligation formulas are Boolean combinations
of components $\mathbb{Q}\Phi$, where $\Phi$ is a finite-trace formula
and $\mathbb{Q}\in\{\exists,\forall\}$ a guarantee or safety trace quantifier.
We first translate individual components into equivalent $\omega$-automata.

Given an \LTLf formula $\Phi$ over $AP$, let 
$\mathcal{A}_\Phi=(T_\Phi,F)$ with $T_\Phi=(2^{AP},Q,\iota,\delta)$ denote an equivalent DFA of size
$2^{2^{\mathcal{O}(|\varphi|)}}$,
represented symbolically.
Assume $\iota\notin F$ when $\mathbb{Q}=\exists$
and $\iota\in F$ when $\mathbb{Q}=\forall$.

Let ${T}^+_\Phi=(2^{AP},Q,\iota,\delta^+)$ and 
${T}^-_\Phi=(2^{AP},Q,\iota,\delta^-)$
be obtained
from $T_\Phi$ by turning accepting (respectively, rejecting) states into absorbing sinks:
define $\delta^+(q,a)=q$ if $q\in F$ and
$\delta^+(q,a)=\delta(q,a)$ otherwise, and
$\delta^-(q,a)=q$ if $q\notin F$ and 
$\delta^-(q,a)=\delta(q,a)$ otherwise.

\begin{restatable}{lemma}{ftctoaut}
\label{lem:ftc_to_automaton}
Let $\Phi$ be an \LTLf formula and $\mathbb{Q}\in\{\exists,\forall\}$. 
Then there exists a DWA $\mathcal{A}_{\mathbb{Q}\Phi}$ 
equivalent to $\mathbb{Q}\Phi$ such that
\begin{itemize}
\item[--] if $\mathbb{Q}=\exists$, then 
$\mathcal{A}_{\mathbb{Q}\Phi}=(T^+_\Phi,F_\Phi)$;
\item[--] if $\mathbb{Q}=\forall$, then
$\mathcal{A}_{\mathbb{Q}\Phi}=(T^-_\Phi,F_\Phi)$.
\end{itemize}
\end{restatable}

\begin{proof}{\emph{(Sketch)}}
Let $\mathcal{A}_\Phi=(T_\Phi,F_\Phi)$ be a DFA for $\Phi$.  
If $\mathbb{Q}=\exists$, an infinite trace satisfies $\exists\Phi$ iff some finite prefix is accepted by $\mathcal{A}_\Phi$; making accepting states absorbing in $T^+_\Phi$ yields a DWA that accepts exactly those runs that eventually remain in $F_\Phi$.  
If $\mathbb{Q}=\forall$, a trace satisfies $\forall\Phi$ iff every prefix is accepted; making rejecting states absorbing in $T^-_\Phi$ ensures that leaving $F$ is permanent, so acceptance coincides with staying in $F_\Phi$ forever.  
In both constructions, SCCs do not mix accepting and rejecting states due to the absorbing sinks, hence the automata are weak.
\end{proof}

This leads to the following transformation from obligation \LTLfp formulas to weak automata.

\begin{corollary}\label{cor:to_automata}
Every obligation \LTLfp formula $\Psi$ can be translated into
an equivalent DWA
of size $2^{2^{\mathcal{O}(|\Psi|)}}$.
\end{corollary}
\begin{proof}
Translate each component $\mathbb{Q}\Phi$ using Lemma~\ref{lem:ftc_to_automaton} and combine the resulting automata according to the Boolean structure of $\Psi$ via the closure constructions of Lemma~\ref{lem:weak_closure}.
\end{proof}

The automata produced during this construction may be minimized at any stage using Corollary~\ref{cor:DWAmin}.

\paragraph{Structural properties of the constructed DWA.}
The automata constructed for a given obligation formula $\Psi$
via Corollary~\ref{cor:to_automata} have a particular internal structure.

Accepting states in the individual automata for guarantee finite trace components
are turned into accepting sinks. These sink states can be merged, leading to individual automata that have just a single accepting (sink) state and in which all
other states are rejecting. Dually, the individual automata for safety components can be assumed to have just a single rejecting sink state. In both cases, it is not possible to leave the sink state, once it has been reached.

After composing the individual automata using the constructions from Lemma~\ref{lem:weak_closure}, the resulting DWA is partitioned according to acceptance/non-acceptance of the individual automata. Let there be $k$ finite trace components. Then every set $\sigma\subseteq \{1,\ldots,k\}$ identifies a region in the composed automaton where component automata are in accepting states if and only if their index is contained in $\sigma$. Formally, let $Q_\Psi$ denote the state set of the composed DWA. Then every state in $Q_\Psi$ is of the shape 
$(q_1,\ldots,q_k)$ where $q_i$ is a state in $\mathcal{A}_{\mathbb{Q}_i\Phi_i}$.
Put $Q_\sigma=\{(q_1,\ldots,q_k)\in Q\mid q_i\in F_{\Phi_i}\text{ iff } i\in \sigma\}$.
Then the sets $Q_\sigma$ for $\sigma\subseteq \{1,\ldots,k\}$ form a partition of $Q_\Psi$.

Since all states in a single such set agree on acceptance of the individual component automata,
we have that each $Q_\sigma$ is either fully accepting ($Q_\sigma\subseteq F_\Psi$) or fully rejecting
($Q_\sigma\cap F_\Psi=\emptyset$), depending on whether $\sigma$ corresponds to a satisfying valuation of the Boolean formula
structure of $\Psi$. 

Furthermore, for all $i$ such that $\mathbb{Q}_i=\exists$ and $i\in \sigma$, or
$\mathbb{Q}_i=\forall$ and $i\notin\sigma$, all states in $Q_\sigma$ have the same
state as $i$th component -- it is the sink state of the automaton $\mathcal{A}_{\mathbb{Q}_i\Phi_i}$.

Since it is not possible to leave sink states in individual automata,
there is, for every strongly connected component $S$ in $\mathcal{A}_\Psi$, some
$\sigma\subseteq\{1,\ldots,k\}$ such that $S\subseteq Q_\sigma$.
Also, for all $\sigma, \sigma'\subseteq\{1,\ldots,k\}$ such that $\sigma\neq\sigma'$, 
$Q_\sigma$ may be reachable from $Q_{\sigma'}$ or vice versa, but not both.
Hence, reachability partially orders the sets $Q_\sigma$, arranging them in a DAG structure.

This does not provide a full SCC decomposition of $\mathcal{A}_\Psi$ since a single set $Q_\sigma$ 
may consist of several SCCs, which, however, are either all fully accepting or all fully rejecting.

\section{Synthesis via DWA}

We show how the synthesis problem for obligation \LTLfp\
can be reduced to the solution of games over deterministic weak automata (DWA). 
Our procedure builds on Lemma~\ref{lem:ftc_to_automaton} and
Corollary~\ref{cor:to_automata}. We introduce infinite-duration games over graphs, 
present the reduction to DWA games, and discuss several symbolic algorithms for solving such games.
In particular, we recall the classical solutions
via B\"uchi games and co\mbox{-}B\"uchi games
and propose a novel
symbolic algorithm that iteratively solves safety and reachability computations. 
These algorithms all have quadratic runtime. Finally, we show that
the novel algorithm can be made linear by restricting its fixpoint computations to individual SCCs
of the game graph; the resulting algorithm resembles the one described in~\cite{DBLP:conf/cav/AmramBFTVW21}.

Formally, a \emph{DWA} (or \emph{weak}) \emph{game} is a B\"uchi game $G=(A,F)$ that is induced by a DWA. Hence
every SCC of $A$ is either contained in $F$ or does not intersect with $F$. The objective of the system player
is to visit $F$ infinitely often (equivalently, to eventually remain within $F$ forever).

\paragraph{Synthesis as DWA games.}
We reduce the reactive synthesis problem for the obligation fragment of \LTLfp\ to the solution of DWA games.
Consider an input obligation formula $\Psi$ given in positive normal form, expressed as a positive Boolean
formula over $k$ components
$\mathbb{Q}_i\Phi_i$, where $\mathbb{Q}_i\in\{\exists,\forall\}$
and each $\Phi_i$ is an \LTLf\ formula.

The synthesis algorithm transforms $\Psi$ into an equivalent DWA (according to Corollary~\ref{cor:to_automata}) 
and then solves the DWA game induced by this automaton.

\paragraph{Step 1.} 
Convert each component $\mathbb{Q}_i\Phi_i$ into an equivalent DWA
$\mathcal{A}_{\mathbb{Q}_i\Phi_i}$ according to Lemma~\ref{lem:ftc_to_automaton}: Transform $\Phi_i$ into the DFA 
$(T_{\Phi_i}, F_{\Phi_i})$, where
$T_{\Phi_i}=(2^{AP},Q_{\Phi_i},\iota_{\Phi_i},\delta_{\Phi_i})$.
If $\mathbb{Q}_i=\forall$, put $\mathcal{A}_{\mathbb{Q}_i\Phi_i}=(T^-_{\Phi_i},F_{\Phi_i})$;
if $\mathbb{Q}_i=\exists$, put $\mathcal{A}_{\mathbb{Q}_i\Phi_i}=(T^+_{\Phi_i},F_{\Phi_i})$.

\paragraph{Step 2.} 
Compose the automata $\mathcal{A}_{\mathbb{Q}_i\Phi_i}$ according to the Boolean structure
of $\Psi$ (repeatedly using Lemma~\ref{lem:weak_closure}): Construct the DWA $\mathcal{A}_{\Psi}=(T_\Psi,F_\Psi)$ inductively as follows.
\begin{align*}
T_{\Psi_i\land\Psi_j}&=T_{\Psi_i\lor\Psi_j}=T_{\Phi_i}\otimes T_{\Phi_j}\\
F_{\Psi_i\land\Psi_j}&=F_{\Psi_i}\times F_{\Psi_j}\\
F_{\Psi_i\lor\Psi_j}&=F_{\Psi_i}\times Q_{\Phi_j}\cup Q_{\Phi_i}\times F_{\Psi_j}
\end{align*}
At any stage, the (partially) composed DWA can be minimized according to
Corollary~\ref{cor:DWAmin}.

\paragraph{Step 3.} Solve the DWA game induced by $\mathcal{A}_\Psi$. If the system player wins from the initial state 
$(\iota_1,\ldots,\iota_k)$ of $T_\Psi$,
extract a witnessing strategy.

\begin{restatable}{theorem}{synth}
\label{thm:synth}
The synthesis problem for the obligation fragment of \LTLfp\ 
can be decided symbolically via DWA games in \textsc{2ExpTime}.
\end{restatable}

More precisely, consider an input formula $\Psi$
of size $n$.
The constructed DWA game over $\mathcal{A}_\Psi$ has size at most
$n'=2^{2^n}$. 
Notably, this bound depends neither on
the number of finite-trace components in $\Psi$ nor on its Boolean structure, but only on the overall size of $\Psi$.
By Lemma~\ref{lem:weakSCCsolution} below, the DWA game can be solved symbolically in time
$\mathcal{O}({n'})$.

Therefore, the worst-case time complexity for synthesizing obligation properties
matches that of reactive synthesis for LTLf.

\begin{remark}
All steps in the described synthesis algorithm admit symbolic implementations.
\end{remark}

\begin{remark}
For a realizable specification $\Psi$, the corresponding transducer is obtained by extracting a memoryless winning strategy 
for the system player in the DWA game induced by $\mathcal{A}_\Psi$ (see, e.g., the proof of Lemma~\ref{lem:weakSCCsolution}).
\end{remark}

\paragraph{Solution of DWA Games.}

We now present several symbolic algorithms for solving DWA games.

Since DWA games can be viewed as both B\"uchi and co\mbox{-}B\"uchi games with a particular structure,
they can be solved using standard algorithms for these classes. In the worst case, however, this is unnecessarily expensive,
as the best known algorithms for (co-)B\"uchi games have quadratic runtime in the arena size.
Nevertheless, we briefly recall the classical nested fixpoint algorithms. 

Next, we propose an apparently novel algorithm that
solves DWA games via iterative safety and reachability computations,
thereby avoiding the computation of nested fixpoints. 
Although its worst-case runtime is again quadratic, the algorithm may perform better in practice
because the intermediately computed winning regions grow monotonically.

Finally, we recall a specialized algorithm 
for DWA games that decomposes the arena into SCCs and solves the resulting DAG,
applying safety and reachability computations to single SCCs in a bottom-up manner, 
thereby achieving linear worst-case runtime.

We begin by defining monotone operators for the
symbolic computation of one-step strategies in a given
game arena $A=(V,E)$:
\begin{align*}
\Diamond W &= \{v\in V\mid E(v)\cap W\neq \emptyset \}, \\
\Box W &= \{v\in V\mid E(v)\subseteq W\},\\
\mathsf{CPre}_s(W) &= (V_s\cap \Diamond W)\cup (V_e\cap \Box W),\\
\mathsf{CPre}_e(W) &= (V_s\cap \Box W)\cup (V_e\cap \Diamond W),
\end{align*}
for $W\subseteq V$. Here, $\Diamond$ computes the set of nodes with an outgoing edge to the argument set,
while $\Box$ computes the set of nodes whose outgoing edges all lead to the argument set.
Consequently, the controllable predecessor operators $\mathsf{CPre}_s$ and $\mathsf{CPre}_e$
compute the nodes from which the system and environment, respectively, have a one-step strategy
to reach the argument set.

A \emph{symbolic operation} is the evaluation of an expression of the form
$W_1\cap W_2$, $W_1\cup W_2$, or $\neg W_1$ for symbolic sets $W_1,W_2$ (encoded, e.g., by BDDs).
We measure time complexity by the number of symbolic operations.
Under this measure, computing $\mathsf{CPre}_s(W)$
or $\mathsf{CPre}_e(W)$ requires symbolic time $\mathcal{O}(1)$.

Let $f:2^V \to 2^V$ be a monotone function. Its extremal (\emph{least} and \emph{greatest}) fixpoints
are defined as
\begin{align*}
\mu X.\,f(X) &= \{Z\subseteq V\mid f(Z)\subseteq Z\} = f^{|V|}(\emptyset)\\
\nu X.\,f(X) &= \{Z\subseteq V\mid Z\subseteq f(Z)\} = f^{|V|}(V)
\end{align*}
where $f^{i+1}(Z)=f(f^i(Z))$ for $i\geq 0$ and $f^0(Z)=Z$.
Hence, a single extremal fixpoint over $V$ can be computed in symbolic 
time $\mathcal{O}(|V|)$, assuming that $f$ can be evaluated in constant symbolic time.

We recall standard fixpoint constructions for games, and associated results. Define
\begin{align*}
\text{Reach}(W,T)&=\mu X. T\cup (W\cap \mathsf{CPre}_s(X)),\\
\text{Safe}(W,T)&=\nu X. T\cup (W\cap\mathsf{CPre}_s(X)),\\
\text{B\"uchi}(T)&=\nu X.\mu Y. (T\cap \mathsf{CPre}_s(X))\cup \mathsf{CPre}_s(Y),\\
\text{co\mbox{-}B\"uchi}(T)&=\mu X.\nu Y. (F\cap \mathsf{CPre}_s(Y))\cup \mathsf{CPre}_s(X).
\end{align*}
Here, $\text{Reach}(W,T)$ computes the system player's winning region in a game over $W\cup T$ with
the objective to eventually reach $T$, 
while $\text{Safe}(W,T)$ computes the winning region with the objective to
either remain in $W$ forever or eventually reach $T$.
Both sets can be computed in symbolic time $\mathcal{O}(|W|)$.

\begin{lemma}
Let $A$ be an arena and $F$ a set of game nodes.
Then $\text{B\"uchi}(F)$ is the winning region of the system player in the
B\"uchi game $G=(A,F)$, while $\text{co\mbox{-}B\"uchi}(F)$
is the winning region of the system player in the
co\mbox{-}B\"uchi game $G=(A,F)$.
\end{lemma}

\begin{corollary}\label{cor:fpsolution}
B\"uchi games and co\mbox{-}B\"uchi games with $n$ nodes can be solved in symbolic time
$\mathcal{O}(n^2)$.
\end{corollary}

\paragraph{Alternating safety and reachability.}
We now consider Algorithm~\ref{alg:safeReach},
which solves DWA games without computing nested fixpoints.
Although the runtime remains quadratic, the algorithm decouples the least and greatest fixpoint computations used in (co-)B\"uchi solutions.
The algorithm maintains a growing sequence of winning regions $W_i$,
starting from $W_0=\emptyset$. In iteration $i$, it solves: 
\begin{itemize}
\item[--] a game over $F\cup W_{2i}$ with the objective to either stay in $F$ forever
or eventually reach $W_{2i}$, and
\item[--] a reachability game over $V$ with target set $W_{2i+1}$.
\end{itemize}
The corresponding fixpoints are
\begin{align*}
W_{2i+1}&=\text{Safe}(F,W_{2i}) & \text{(line 3)}\\
W_{2i+2}&=\text{Reach}(V,W_{2i+1})& \text{(line 4)}
\end{align*}
Thus, $\textsc{SolveSafeReach}(V,F)$ computes the set of nodes from which
the system player has a strategy to eventually stay within
an accepting SCC forever.

\begin{algorithm}[bt]
\caption{\label{alg:safeReach}$
\textsc{SolveSafeReach}($V$,$F$)$}
\SetKwFunction{Safe}{Safe}
\SetKwFunction{Reach}{Reach}
\DontPrintSemicolon
   $i=0$;\,$W_0=\emptyset$;\,$W_{-2}=V$\;
   \While{$W_{2i}\neq W_{2(i-1)}$}{
    $W_{2i+1}=\Safe(F,W_{2i})$\;
    $W_{2i+2}=\Reach(V,W_{2i+1})$\;
    $i=i+1$\;
   }
   \Return{$W$};

\end{algorithm}

We observe that $W_{2i}\subseteq W_{2i+1}\subseteq W_{2(i+1)}$ for all $i$.
Consequently, the fixpoint computations in later iterations 
operate on larger target sets ($W_{2i}$ and $W_{2i+1}$, respectively) and therefore
terminate more quickly.

\begin{restatable}{lemma}{safereach}
\label{lem:safereachsolution}
Algorithm~\ref{alg:safeReach} solves weak games with $n$ nodes, $k$ accepting nodes and $l$ SCCs 
in symbolic time $\mathcal{O}(n+kl)\in \mathcal{O}(n^2)$, and yields memoryless winning strategies for the system player.
\end{restatable}

\begin{proof} Correctness follows directly from the construction of the algorithm.
The total number of symbolic operations required to compute all sets $\text{Reach}(V,W_j)$ is
linear in $n$, since each game node is added to the winning region at most once.
Furthermore, at most $l$ computations of $\text{Safe}(F,W_j)$ are needed; each such computation
can be implemented in symbolic time $\mathcal{O}(k)$, where $k=|F|$.
A winning strategy is obtained by playing, at each game node $v$, according to the memoryless 
winning strategy for the game associated with $W_j$, where $j$ is the smallest index such that $v\in W_j$.
\end{proof}

\begin{algorithm}[bt]
\caption{\label{alg:weakSCC}
\textsc{SolveWeakSCC}($V$,$F$,$\mathrm{SCCs}$)}
\SetKwFunction{Safe}{Safe}
\SetKwFunction{Reach}{Reach}
\SetKwFunction{bottomSCCs}{bottomSCCs}
\DontPrintSemicolon
   $W=\emptyset$\;
   \While{$SCCs\ne \emptyset$}{
   $B=\bottomSCCs(SCCs)$\;
   $SCCs=SCCs\setminus B$\;
   \For{$SCC \in B$}{
    \eIf{$SCC\subseteq F$}{
      $X=\Safe($SCC, W$)$\;
    }(\tcp*[h]{$SCC\cap F=\emptyset$}){
      $X=\Reach($SCC, W$)$ 
    }
    $W=W\cup X$\;
    }
    }
   \Return{$W$}
\end{algorithm}
\paragraph{Linear solution of weak games.}
Algorithm~\ref{alg:safeReach} can be made linear
by parameterizing it with an SCC decomposition of the game
arena and restricting fixpoint computations to individual SCCs.
This results in Algorithm~\ref{alg:weakSCC}, which 
is similar to the algorithm proposed in~\cite{DBLP:conf/cav/AmramBFTVW21}. 
It takes a game together with its SCC decomposition as input and solves the
DAG of SCCs in a bottom-up fashion by repeatedly processing all bottom SCCs and removing them from the DAG.
The algorithm assumes a function
$\textsc{bottomSCCs}$ that returns the current minimal SCCs.
The winning region computed so far is stored in the set $W$.

Each bottom SCC is handled as follows.
If the SCC is accepting, the algorithm solves a game over
the union of the SCC and $W$ 
with the objective to either remain within the SCC forever or eventually
reach $W$. If the SCC is rejecting,
the objective is just to eventually reach $W$.

\begin{restatable}{lemma}{scc}
\label{lem:weakSCCsolution}
Algorithm~\ref{alg:weakSCC} solves weak games with $n$ nodes 
in symbolic time $\mathcal{O}(n)$ and yields memoryless winning strategies for the system player.
\end{restatable}

\begin{proof}{\emph{(Sketch)}}
The algorithm computes precisely the set of nodes for which the system player
has a strategy to eventually stay in some accepting SCC forever. 
Regarding time complexity, graphs with $n$ vertices can be decomposed into their SCCs in symbolic time 
$\mathcal{O}(n)$~\cite{DBLP:conf/tacas/LarsenSSJPP23}.
For each strongly connected component $SCC$, Algorithm~\ref{alg:weakSCC}
computes exactly one of the sets
$\text{Reach($SCC,W$)}$ or $\text{Safe($SCC,W$)}$, each of which
can be obtained in symbolic time $\mathcal{O}(|SCC|)$. Since the sizes of all SCCs 
sum to $n$, the overall runtime is linear.
A memoryless strategy is constructed
by always following the memoryless safety or reachability strategy for the current SCC. 
\end{proof}

\section{Implementation}\label{sec:impl}
We implemented all the algorithms mentioned above by extending the \LTLfp synthesis tool \tool~\cite{10.24963/kr.2025/78}.

We first describe the implementation of the game arena construction and then detail the implementation of the game-solving algorithms.
Given an obligation \LTLfp formula $\Psi$, each finite-trace component $\mathbb{Q}\Phi$ is converted into
a DWA by translating the \LTLf formula $\Phi$ into a DFA using MONA's explicit
representation~\cite{DBLP:journals/ijfcs/KlarlundMS02}, which then is treated as a DWA.
Subsequently, the composition constructions described in~\Cref{sec:obltodwa} are applied along the Boolean structure of $\Psi$.
We implemented DWA minimization in an explicit representation, applying L\"oding's pre-processing step~\cite{DBLP:journals/ipl/Loding01}, followed by DFA minimization using MONA.
We consider two modes for obtaining the automata:
\begin{itemize}
    \item \textbf{Component-wise minimization + symbolic product}: Each component DWA is minimized individually,
    after which the product is constructed symbolically. This approach results in $n$ minimization calls for $n$ components.
    \item \textbf{Incremental explicit products + threshold-based switching:} Intermediate products are kept explicit and are minimized as long as the number of states remains below a fixed threshold $\tau$ (we use $\tau = 256$).
    Once the threshold is exceeded, we switch to a symbolic representation for subsequent products. To further control
    the sizes of product automata, we compute products in a balanced manner.
\end{itemize}
At the end of the arena construction, both variants yield symbolic games played over DWA. These games are solved as follows.

\paragraph{Direct fixpoint computation algorithms.} We implemented the classical fixpoint computations
for the symbolic solution of B\"uchi and co-B\"uchi games, as we well as the novel Algorithm~\ref{alg:safeReach}
(referred to as ``SafeReach'' below). These algorithms integrate seamlessly into the \tool infrastructure, as
they compute fixpoints directly over the BDD representation of the game arena. 

\paragraph{SCC-based solution algorithm.} Although Algorithm~\ref{alg:weakSCC} is theoretically more efficient, 
implementing linear-time SCC decomposition is not straightforward with the automaton representation used by \tool.
Existing linear-time algorithms assume a monolithic, fully symbolic encoding of the transition relation with primed successor variables (i.e., transitions encoded as a function $T(x, x')$). In contrast, \tool employs a partitioned encoding that facilitates pre-image computation but makes post-image computation more difficult.
In our implementation of SCC decomposition, we follow the established FwdBwd approach, using an implementation similar to~\cite{DBLP:conf/cav/AmramBFTVW21}.
Specifically, we first convert the transition relation to the primed encoding and then compute the path relation required for SCC decomposition via transitive closure.
We also considered implementing the fully symbolic Chain algorithm~\cite{DBLP:conf/tacas/LarsenSSJPP23}. However, this algorithm optimizes only the SCC computation itself. This does not address the overhead of constructing the monolithic transition relation.
Since this construction constitutes the dominant bottleneck in our setting, we do not expect the Chain algorithm to significantly improve scalability.

\begin{figure}[t]
    \centering
    \includegraphics[width=\linewidth]{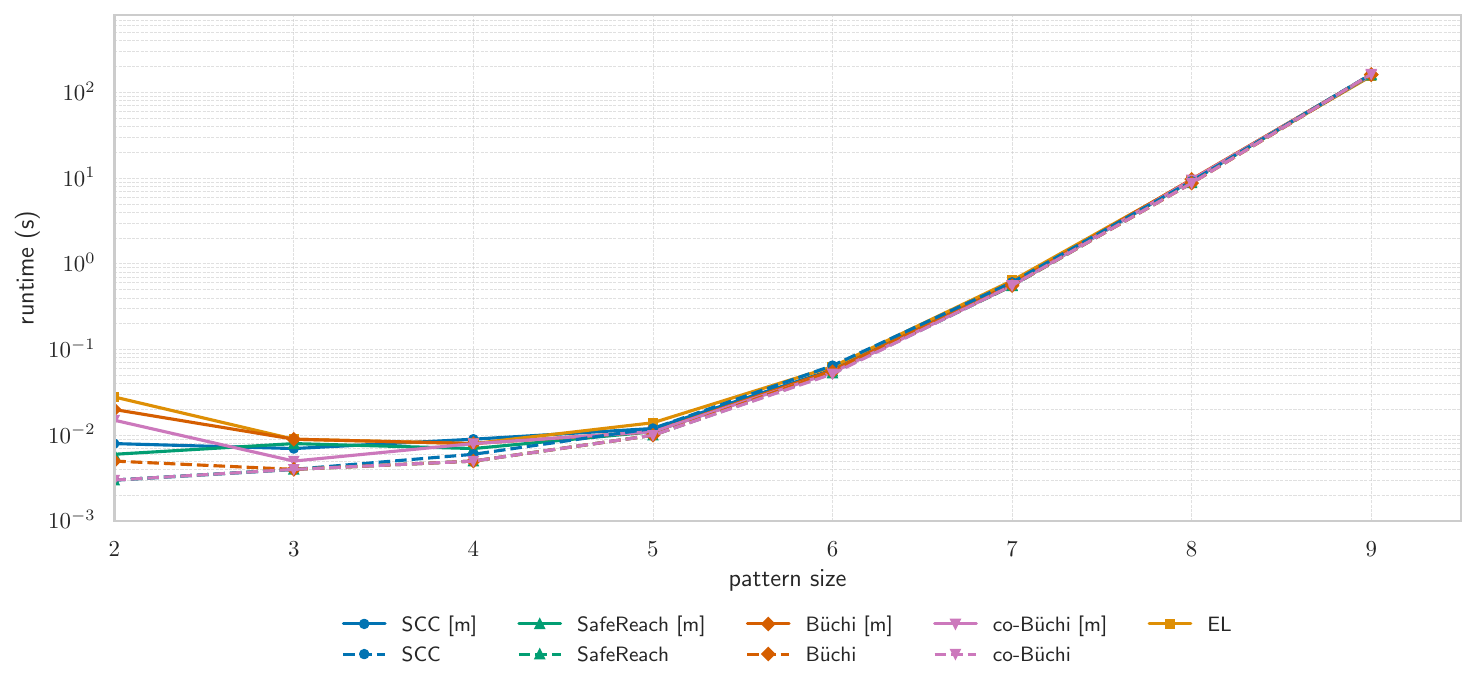}
    \caption{Runtime of solvers on the counter pattern $\psi$}
    \label{fig:counter}
\end{figure}

\begin{figure*}[t]
    \centering

    \begin{subfigure}{0.48\textwidth}        \includegraphics[width=\linewidth]{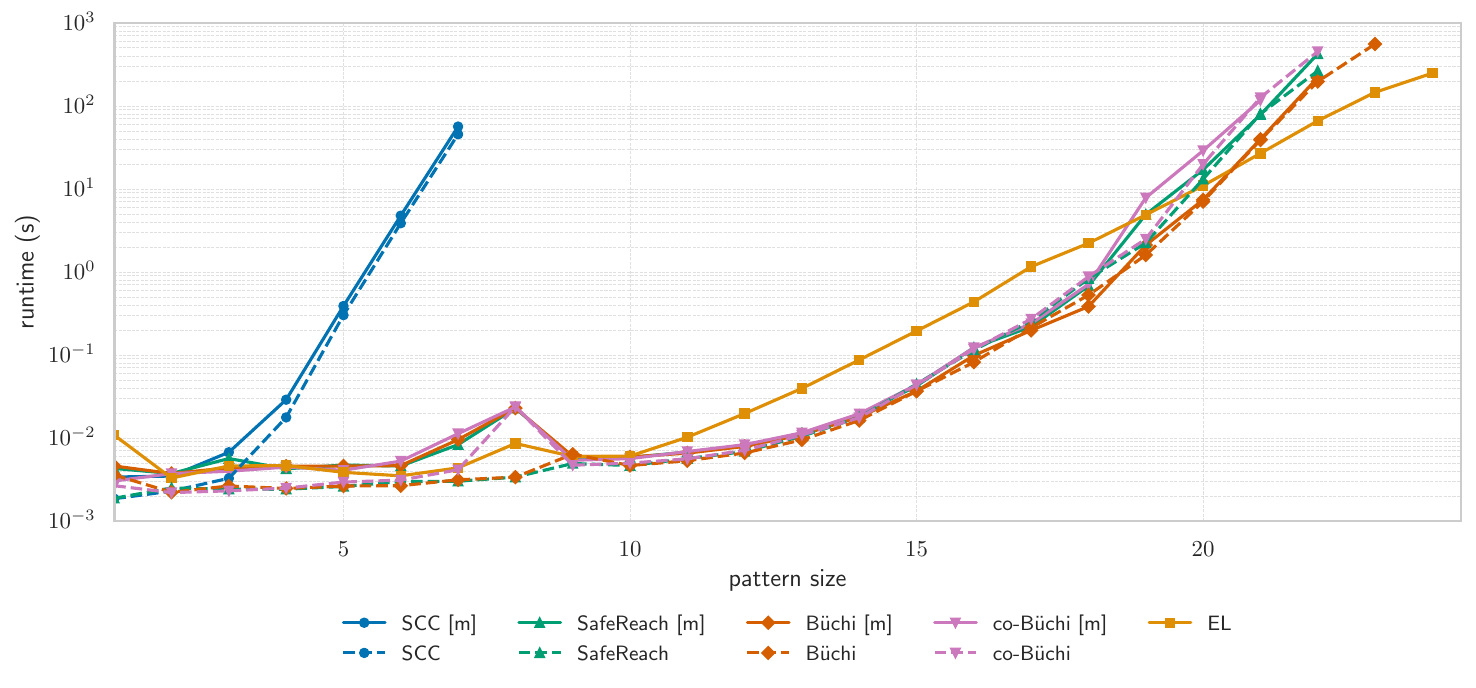}
        \caption{$\textstyle\bigwedge_i{\exists_i}\land\forall$}
        \label{fig:aAe}
    \end{subfigure}\hfill
    \begin{subfigure}{0.48\textwidth}
        \includegraphics[width=\linewidth]{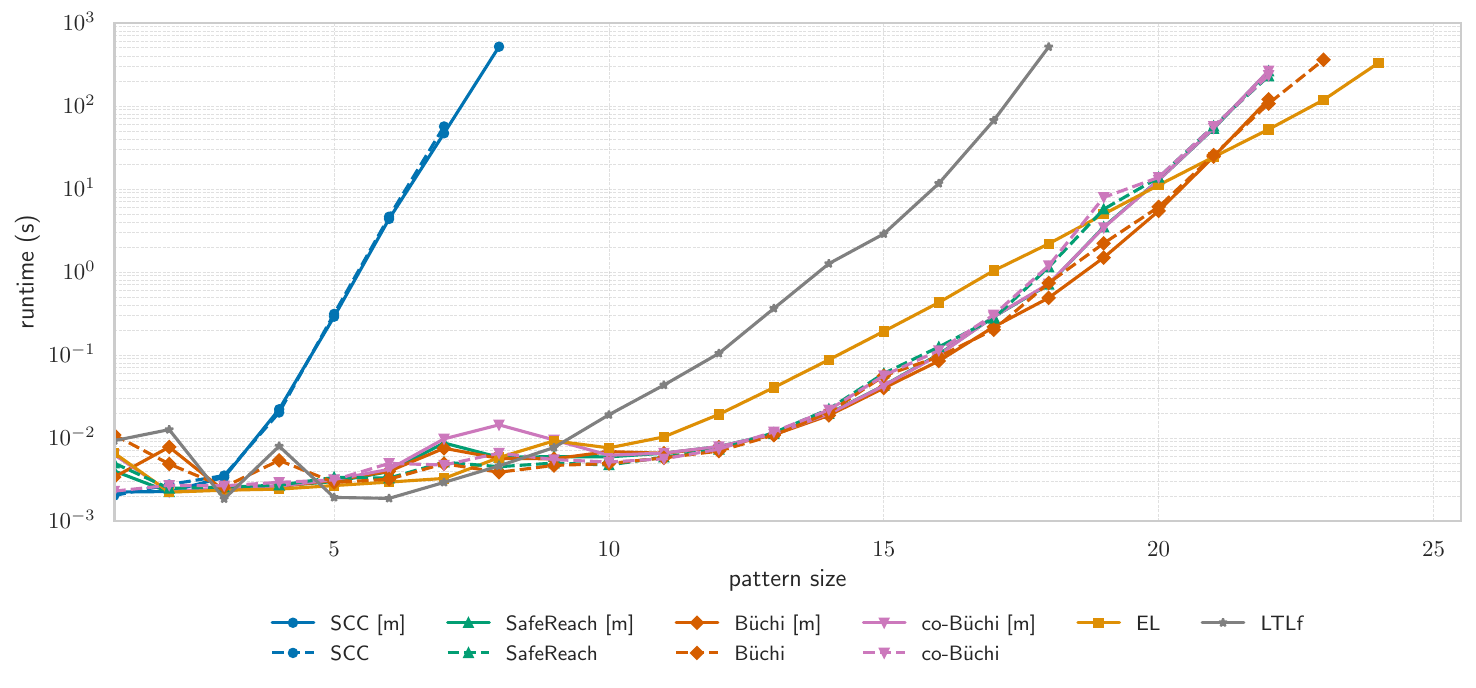}
        \caption{$\textstyle\bigwedge_i{\exists_i}\land\exists$}
        \label{fig:eAe}
    \end{subfigure}

    \medskip

    \begin{subfigure}{0.48\textwidth}
        \includegraphics[width=\linewidth]{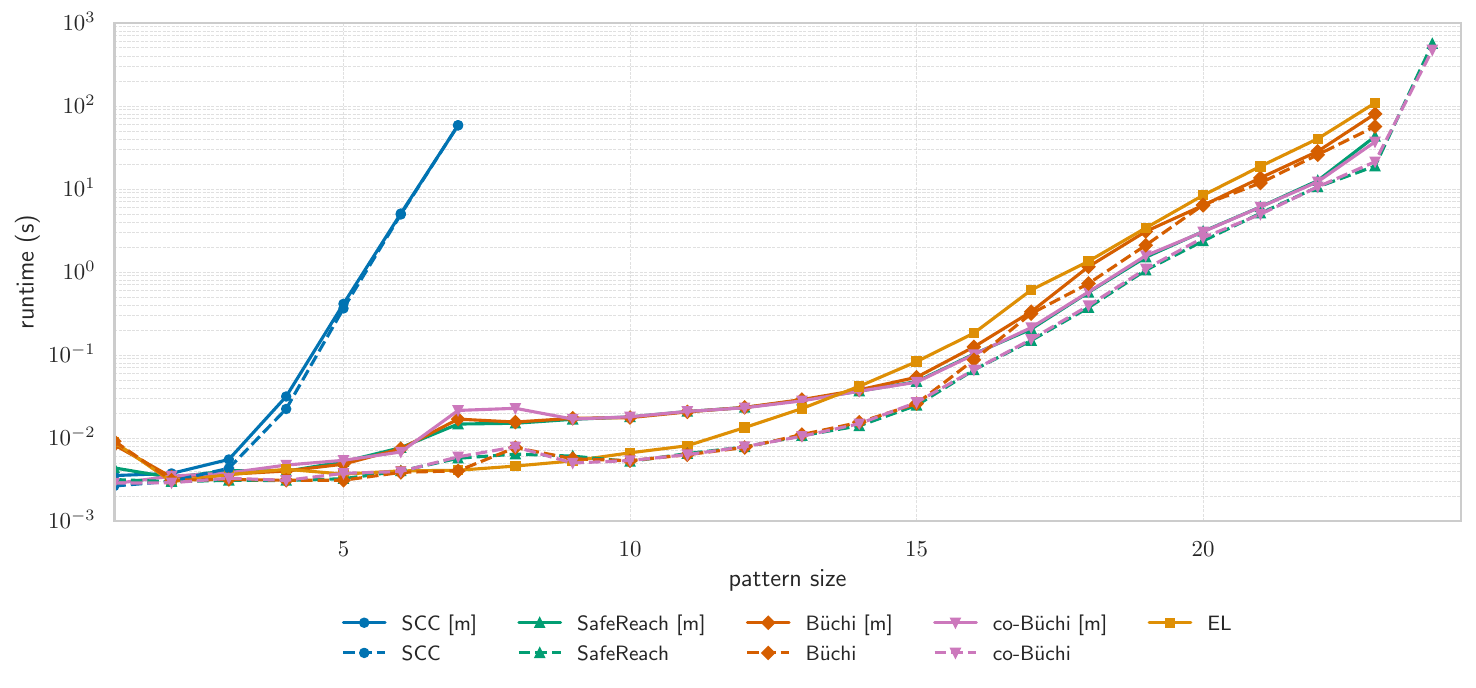}
        \caption{$\textstyle\bigvee_i{\forall_i}\lor\forall$}
        \label{fig:aOa}
    \end{subfigure}\hfill
    \begin{subfigure}{0.48\textwidth}
        \includegraphics[width=\linewidth]{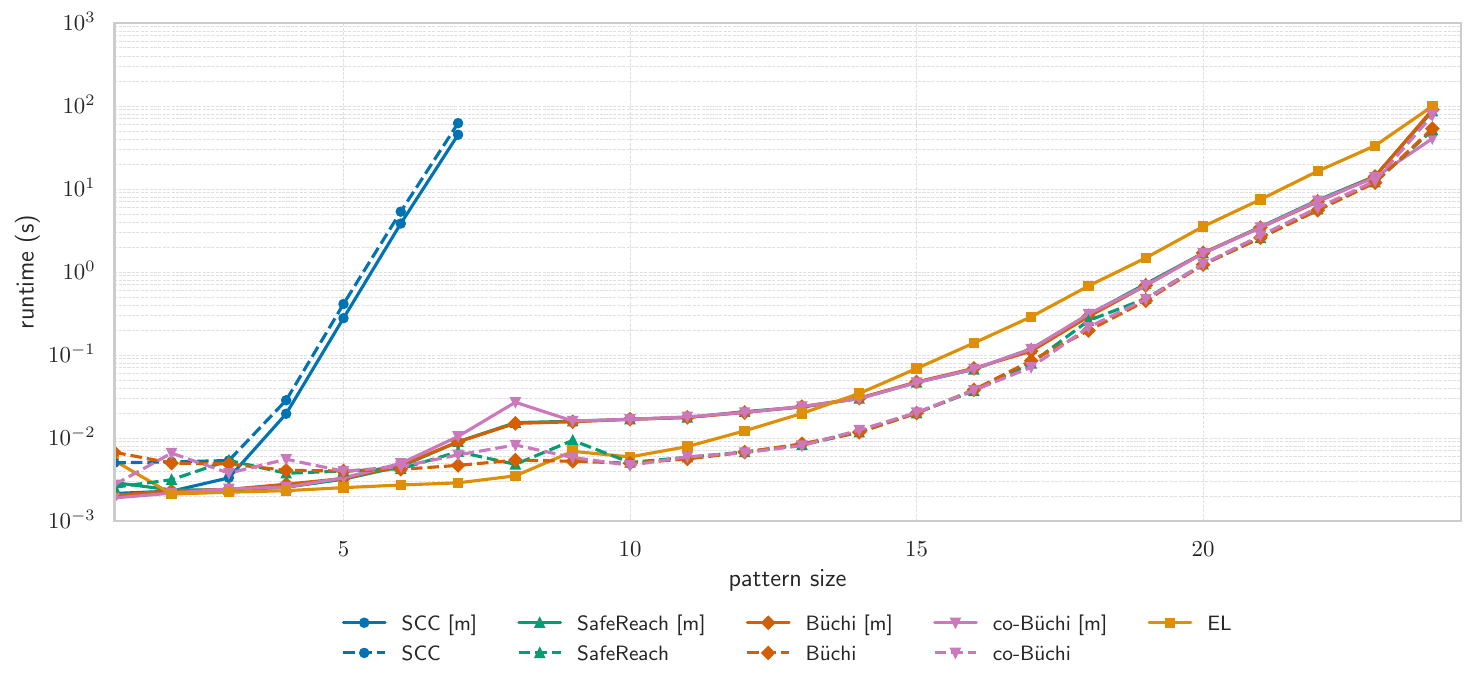}
        \caption{$\textstyle\bigvee_i{\forall_i}\lor\exists$}
        \label{fig:aOe}
    \end{subfigure}

    \caption{Runtime of solvers for conjunction and disjunction patterns}
    \label{fig:patterns}
\end{figure*}

\section{Experiments}\label{sec:bench}
We conduct experiments on a range of formula patterns to study the scalability of the proposed algorithms. The evaluation is intended as a proof of concept and therefore does not rely on extensive computational resources. Overall, our results demonstrate that the obligation fragment can be solved with effectivity comparable to standard \LTLf synthesis. All experiments were carried out on an M4 MacBook Air with 16GB RAM of which 12GB were allocated to the virtual experiments, using a time limit of 10 minutes per benchmark. 

When discussing experiment results, we refer to the four synthesis algorithms detailed above by
``B\"uchi'', ``co-B\"uchi'', ``SafeReach'', and ``SCC''. Versions that employ incremental minimization
are marked with ``[m]'' and plotted using solid lines (dashed lines are used to plot the runtime for algorithms that build the component-wise
symbolic product of minimized automata). 
We additionally compare our methods with the \tool implementation of synthesis for full \LTLfp based on a reduction to Emerson-Lei games
(denoted ``EL'' in the plots); this algorithm does not support incremental minimization.
Finally, we compare against the \tool implementation of \LTLf synthesis (denoted by ``LTLf''), where possible. To ensure a fair comparison, we disabled one-step realizability checks using Z3, as these trivially solve the instance.

We begin by adapting the counter benchmark commonly used in \LTLf{}
and \LTLfp synthesis, which specifies a binary counter with $n$ bits $b_0,\ldots, b_{n-1}$ with the overall
system objective to reach the maximal counter value, expressed as a guarantee property. 
The original \LTLfp formulation~\cite{10.24963/kr.2025/78} does not belong to the obligation fragment, 
as it employs a recurrence property to require that the counter is incremented
infinitely often. We modify this specification by incorporating the increment requirement
into the \LTLf formula for the domain. More specifically, we require that the counter is incremented at every step ($\varphi_{alw}$), resulting in the following realizable specification in which the system controls the variable $add$:
\begin{align*} \scriptsize
    \psi &= \forall (\varphi_{init} \wedge \varphi_{inc} \wedge \varphi_{alw} \wedge \varphi_{trans}) \supset \exists \varphi_{goal} \\
    \varphi_{init} &= \lnot c_0 \wedge \dots \wedge \lnot c_{n-1} \wedge \lnot b_0 \wedge \dots \wedge \lnot b_{n-1} \\
    \varphi_{inc} &= \mathsf{G}(add\supset (\mathsf{X}(c_0) \wedge \mathsf{X}(\mathsf{X}(c_0)) \wedge \mathsf{X}(\mathsf{X}(\mathsf{X}(c_0))))) \\
    \varphi_{alw} &= \mathsf{G}\mathsf{F}(add \wedge \mathsf{X}\,\mathsf{false}) \\
    \varphi_{trans} &= \mathsf{G} \left\{\begin{aligned}
((\lnot c_i \land \lnot b_i) &\to \mathsf{X}(\lnot b_i \land \lnot c_{i+1})) \wedge \\
((\lnot c_i \land  b_i) &\to \mathsf{X}(b_i \land \lnot c_{i+1})) \wedge \\
((c_i \land \lnot b_i) &\to \mathsf{X}(b_i \land \lnot  c_{i+1})) \wedge \\
((c_i \land  b_i) &\to \mathsf{X}(\lnot b_i \land c_{i+1}))
\end{aligned}\right. \\
\varphi_{goal} &= \mathsf{F}\left( b_0 \land \dots \land b_{n-1} \wedge \mathsf{X}\,\mathsf{false} \right)
\end{align*}
\Cref{fig:counter} shows the runtimes of the various implementations
on this benchmark series.
In this case, DFA construction accounts for the majority of the runtime, while solving the resulting game is comparatively fast.
Consequently, the overall runtime is dominated by the construction of the \LTLf component automata, 
leading to nearly identical performance across all algorithms. 
Minimization does not significantly reduce the number of states in this benchmark. This appears to be due primarily to the minimization threshold being set to 256: once the automaton for the \LTLf domain exceeds this size, minimization is no longer performed. 
Increasing the DWA minimization threshold is unlikely to yield significant performance gains, as DFA construction already constitutes the principal cost in this example.

To evaluate synthesis for basic Boolean combinations of finite-trace components, we consider several formula patterns obtained by combining \LTLf formulas of the form $\varphi_i=\mathsf{F}( (e_i \vee a_i) \wedge \mathsf{X}\,\mathsf{false})$ with different boolean operators and trace quantifiers. The environment controls variables $e_i$ and the system
controls variables $a_i$. In consequence, all formulas in this experiment are realizable. We omit patterns in which the number of finite-trace components can be reduced using \Cref{lem:simp}. 
\[
\scriptsize
\begin{aligned}
\textstyle\bigwedge_i{\exists_i}\wedge\forall
  &:\qquad \exists\varphi_1 \land \dots \land \exists\varphi_{n-1} \land \forall\varphi_n \\
\textstyle\bigwedge_i{\exists_i}\wedge\exists
  &:\qquad \exists\varphi_1 \land \dots \land \exists\varphi_n \\
\textstyle\bigvee_i{\forall_i}\vee\forall
  &:\qquad \forall\varphi_1 \vee \dots \vee \forall\varphi_n \\
\textstyle\bigvee_i{\forall_i}\vee\exists
  &:\qquad \forall\varphi_1 \vee \dots \vee \forall\varphi_{n-1} \vee \exists\varphi_n
\end{aligned}
\]
The ${\exists}$ pattern introduced in \cite{10.24963/kr.2025/78} corresponds to the $\textstyle\bigwedge_i{\exists_i}\wedge\exists$ pattern. We refer to the number of conjuncts or disjuncts as the pattern size.

\Cref{fig:patterns} shows the runtime of the algorithms on these benchmarks for growing pattern sizes.

In most of these examples, we observe a clear separation between the SCC, EL, and B\"uchi-based algorithms, both in terms of runtime and the number of instances solved. 
With the exception of the SCC-based implementation, our new methods perform similar to
the existing algorithms. 
We attribute the comparatively good performance of the EL-based solver for 
experiments on conjunctions of guarantuee properties (Figures~\ref{fig:patterns}(a) and~\ref{fig:patterns}(b)) to
the fact that the EL solver constructs and solves generalized B\"uchi games in these experiments. Using this approach, the EL solver decomposes the overall objective into several subobjectives that can be solved independently. An overall solution is obtained by composing the solutions for the subobjectives. This apparently can lead to advantages over the B\"uchi-based algorithms which solve games with a single monolithic objective, which is simpler but may require more iterations until a fixpoint is obtained.

Because the patterns in this experiment are designed to be irreducible under the simplification rules from Lemma~\ref{lem:simp}, it is not surprising that minimization in most cases does not substantially reduce the number of states. The only exception is the $\textstyle\bigwedge_i{\exists_i}\wedge\forall$ pattern, where minimization decreases the state count by up to 50\%, but does not reduce the bit count of the BDD representation.

We attribute the comparatively weak performance of the SCC-based algorithm primarily to implementation effects, in particular the automata representation used in \tool, which is not well suited for efficient SCC computation.
In instances where the SCC solver fails to produce a result within the allotted time, the construction of the monolithic transition relation typically fails because the BDD grows too large. 
Notably, minimization has mixed effects on the SCC algorithm: in some cases (e.g., the $\textstyle\bigvee{\forall_i}\vee\forall$ pattern) it improves performance, while in others (e.g., the $\textstyle\bigwedge_i{\exists_i}\wedge\exists$ pattern) it degrades it. This likely reflects a trade-off between the reduced state space, which can improve efficiency, and the loss of the partitioned symbolic structure of the automata caused by minimization, which can make subsequent operations more expensive.

The $\textstyle\bigwedge_i{\exists_i}\wedge\exists$ pattern is equi-realizable with the \LTLf formula $\bigwedge_{1 \leq i \leq n} (a_i \vee e_i)$, enabling a direct comparison with \LTLf synthesis tools. In \Cref{fig:eAe} we plot the runtime of the \LTLf synthesizer included in \tool. 
We observe that both the EL-based approach and our novel algorithms require less time and solve more instances. 
We believe that the improved performance of our implementations in this benchmark stems from their compositional treatment of the automata. 
To validate this hypothesis, we conduct additional experiments using our solvers on formulas $\exists (\bigwedge_{1 \leq i \leq n} \mathsf{F((a_i \vee e_i) \wedge \mathsf{X}~\mathsf{false})})$. 
This results in performance comparable to standard \LTLf synthesis, suggesting that a compositional approach to \LTLf synthesis ~\cite{bansal2020hybrid} could achieve similar results on these benchmarks.

Overall, the results in \Cref{fig:eAe} are particularly significant, as they show that synthesis for the \LTLfp{} obligation fragment, despite operating over infinite traces, achieves performance comparable to that of \LTLf synthesis on finite traces.

Finally, we construct a family of benchmark formulas with a more complex Boolean structure.
We consider the implication pattern
$\psi_j = \textstyle \bigwedge_{i \leq j} ((\exists \mathsf{F}\,a_i) \supset \exists (\mathsf{F}\,e_i))$
 which states
that for every $i\leq j$, whenever $a_i$ holds at some point along a trace, $e_i$ must also hold at some point 
(possibly even before $a_i$). The environment controls the variables $a_i$, while the system controls the variables $e_i$,
making the overall specification realizable.
Figure~\ref{fig:impl} shows the runtimes for this experiment. In this setting, the SCC algorithm benefits
from minimization, which reduces the BDD representation by four bits for the larger instances and by
two to three bits for pattern sizes up to 14. Nevertheless, its performance remains inferior to that of the EL algorithm, which itself is outperformed significantly by our B\"uchi-based approaches. 
For the B\"uchi algorithms, minimization appears to have only limited impact. Although the arena size reduces by roughly 99\% on the larger instances, this results in a reduction of only about 12\% in the BDD bit count, which likely explains the modest effect on overall runtime.

\begin{figure}
    \centering
    \includegraphics[width=\linewidth]{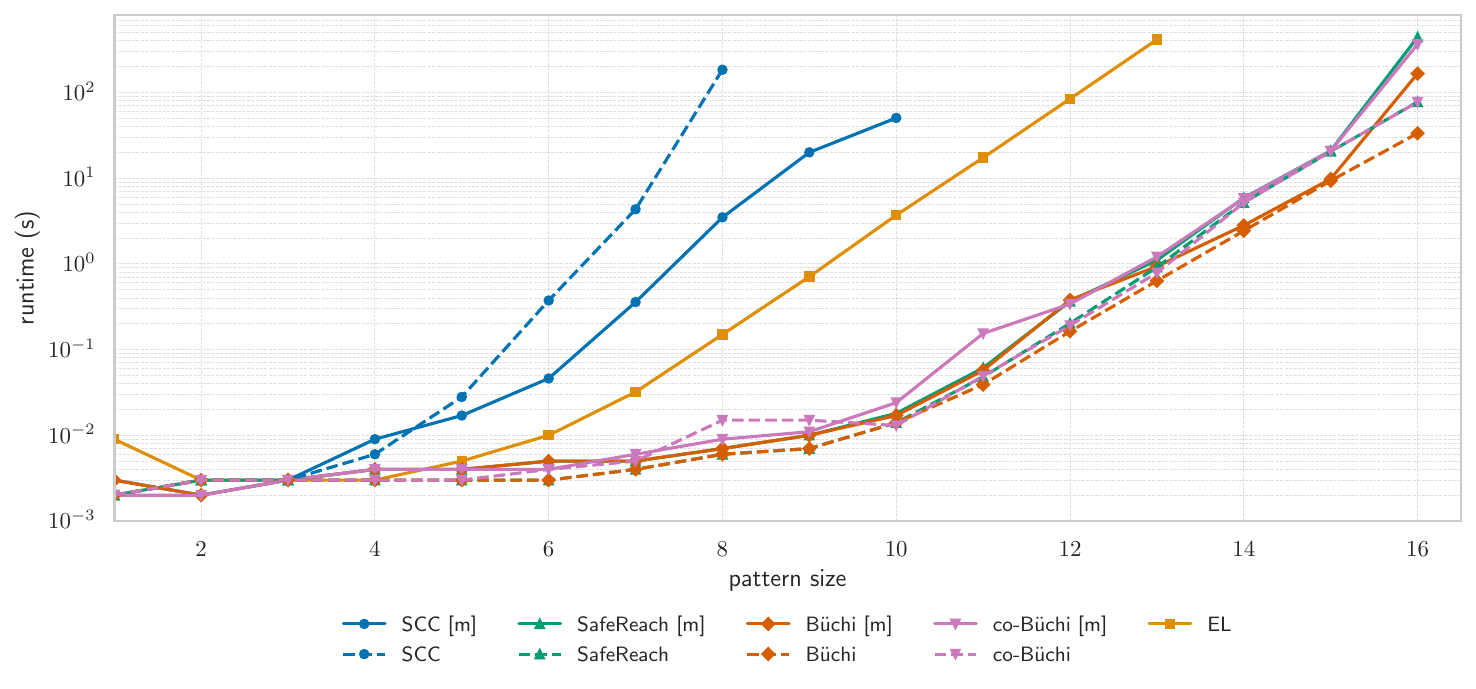}
    \caption{Runtime of solvers on the implication pattern}
    \label{fig:impl}
\end{figure}

\section{Discussion}

We show that synthesis for obligation properties in \LTLfp reduces to solving games over automata on infinite words with weak acceptance conditions (DWA). Leveraging the simplicity of DWA automata and games, we propose reactive synthesis algorithms with the same asymptotic complexity as \LTLf synthesis. We implement and evaluate these algorithms on benchmark formulas, comparing them to each other, to a full \LTLfp solver based on reduction to Emerson–Lei games, and, where applicable, to an \LTLf synthesis tool. Results show performance comparable to \LTLf synthesis, extending effective synthesis from finite traces to infinite-trace obligation properties.

For future work, we plan to implement the SCC-based algorithm using an automaton representation akin to SPOT’s MTBDD-based approach~\cite{DuretLutzZPGV25}. We also conjecture that our efficient methods can be adapted to handle recurrence and persistence properties, but leave this for future investigation.

\bibliographystyle{kr}
\bibliography{lib,gdg-references}

\clearpage

\section*{Supplement to ``Preliminaries''}

\paragraph{Syntax and Semantics of \LTL.}
Linear-time Temporal Logic (\LTL)~\cite{Pnueli77} allows to express temporal properties of infinite traces.
The set of \LTL formulas over the set $AP$ of atomic propositions is 
given by the following grammar.
\begin{align*}
\varphi,\psi::= p \mid \neg\varphi\mid \varphi\land\psi\mid \mathsf{X} \varphi \mid \varphi\mathsf{U}\psi \tag{$p\in AP$}
\end{align*}
We make use of common abbreviations such as $\varphi\lor\psi=\neg(\neg \varphi\land\neg\psi)$, $\mathsf{true}=p\lor \neg p$, $\mathsf{false}=\neg\mathsf{true}$, $\mathsf{F}\varphi=\mathsf{true}\mathsf{U}\varphi$ (``eventually'') and
$\mathsf{G}\varphi=\neg \mathsf{F}\neg\varphi$ (``always'').
\LTL formulas are evaluated over infinite traces $\tau\in (2^{AP})^\omega$ of sets of atomic propositions. 
Satisfaction of \LTL formulas by infinite traces is defined inductively as follows, referring to the $i$th element in $\tau$ by $\tau_i$.
\begin{align*}
\tau,i&\models p &\text{ iff } &p\in \tau_i\\
\tau,i&\models \neg\varphi &\text{ iff } &\tau,i\not\models\varphi\\
\tau,i&\models \varphi\land\psi &\text{ iff } &\tau,i\models \varphi \text{ and } \tau,i\models \psi\\
\tau,i&\models \mathsf{X}\varphi &\text{ iff } &\tau,i+1\models\varphi\\
\tau,i&\models \varphi\mathsf{U}\psi &\text{ iff } &\exists j\geq i.\,\tau,j\models \psi \text{ and } \\
&&&\forall i\leq j'<j.\tau,j'\models \varphi.
\end{align*}
Given an \LTL formula $\varphi$, we let 
\[
[\varphi]=\{\tau\in (2^{AP})^\omega\mid \tau,0\models\varphi\}
\]
denote
the set of infinite traces that satisfy $\varphi$ at the start.

\paragraph{Syntax and Semantics of \LTLf.}
The syntax of \LTL on finite traces (\LTLf)~\cite{DegVa13} is the same as the syntax of \LTL given above. However, \LTLf formulas are evaluated
over finite traces $\tau\in (2^{AP})^*$ rather than over infinite ones. A common abbreviation is the operator $\mathsf{X}\varphi=\neg \mathsf{X}[!]\neg\varphi$ (``weak next''),
expressing that if there is a next position in the trace, then it satisfies $\varphi$.
The satisfaction of \LTLf formulas by finite traces is defined inductively
as follows, where $|\tau|$ denotes the length of a finite trace $\tau$.
\begin{align*}
\tau,i&\models p &\text{ iff } &p\in \tau_i\\
\tau,i&\models \neg\varphi &\text{ iff } &\tau,i\not\models\varphi\\
\tau,i&\models \varphi\land\psi &\text{ iff } &\tau,i\models \varphi \text{ and } \tau,i\models \psi\\
\tau,i&\models \mathsf{X[!]}\varphi &\text{ iff } &i+1< |\tau|\text{ and }\tau,i+1\models\varphi\\
\tau,i&\models \varphi\mathsf{U}\psi &\text{ iff } &\exists i\leq j<|\tau|.\,\tau,j\models \psi \text{ and } \\
&&&\forall i\leq j'<j.\tau,j'\models \varphi.
\end{align*}
The subtle (but consequential) difference to standard \LTL semantics is the requirement that $i+1<|\tau|$ (resp. $j<|\tau|$) in the last two clauses;
that is, for all $\varphi$, $\mathsf{X[!]}\varphi$ is not satisfied at the end of a finite trace, and in
order for $\varphi\mathsf{U}\psi$ to be satisfied in a finite trace, $\psi$ is required to be satisfied before the trace ends.
The formula $\mathsf{last} = \neg \mathsf{X[!]}\,\mathsf{true}$ is satisfied exactly at the last position of a finite trace.

Given an \LTLf formula $\varphi$, we let 
\[
[\varphi]=\{\tau\in (2^{AP})^*\mid \tau,0\models\varphi\}
\]
denote the set of \emph{finite} traces that satisfy $\varphi$ at the start.

\paragraph{Syntax and Semantics of \PPLTL.}

Pure past LTL (\PPLTL)~\cite{DeGiacomoSFR20} allows to express temporal properties of finite traces
by means of statements that refer to the past.
The set of \PPLTL formulas over the set $AP$ of atomic propositions is 
given by the following grammar.
\begin{align*}
\varphi,\psi::= p \mid \neg\varphi\mid \varphi\land\psi\mid \mathsf{Y[!]} \varphi \mid \varphi\mathsf{S}\psi \tag{$p\in AP$}
\end{align*}
Here, $\mathsf{Y[!]}$ (``yesterday'') and $\mathsf{S}$ (``since'') are the past operators; they
are past analogues of the temporal operators $\mathsf{X[!]}$ and $\mathsf{U}$.
Common abbreviations include $\mathsf{O}\varphi=\mathsf{true}\,\mathsf{S}\,\varphi$ (``at least once in the past'') and
$\mathsf{H}\varphi=\neg \mathsf{O}\neg\varphi$ (``historically''). 
The satisfaction of \PPLTL formulas by finite traces is defined inductively
as follows.
\begin{align*}
\tau,i&\models p &\text{ iff } &p\in \tau_i\\
\tau,i&\models \neg\varphi &\text{ iff } &\tau,i\not\models\varphi\\
\tau,i&\models \varphi\land\psi &\text{ iff } &\tau,i\models \varphi \text{ and } \tau,i\models \psi\\
\tau,i&\models \mathsf{Y[!]}\varphi &\text{ iff } &i>0\text{ and }\tau,i-1\models\varphi\\
\tau,i&\models \varphi\,\mathsf{S}\,\psi &\text{ iff } &\exists i\geq j\geq 0.\,\tau,j\models \psi \text{ and } \\
&&&\forall i\geq j'>j.\tau,j'\models \varphi.
\end{align*}
We observe that for all $\varphi$, $\mathsf{Y[!]}\varphi$ is not satisfied at the start of a finite trace, and in
order for $\varphi\,\mathsf{S}\,\psi$ to be satisfied at position $i$ in a finite trace, $\psi$ has to be satisfied somewhere in the first $i+1$ positions
of the trace.
The formula $\mathsf{first} = \mathsf{Y}\,\mathsf{false}$ (where $\mathsf{Y}=\neg \mathsf{Y}\mathsf[!]\neg$ denotes~``weak yesterday") is satisfied exactly at the first position of a finite trace.

Given a \PPLTL formula $\varphi$, we let 
\[
[\varphi]=\{\tau\in (2^{AP})^*\mid \tau,|\tau|-1\models\varphi\}
\]
denote the set of \emph{finite} traces that satisfy $\varphi$ at the \emph{end}, that is, at position $|\tau|-1$.

\section*{Supplement to ``Specification of Obligation Properties''}

\paragraph{Proof of Lemma~\ref{lem:simp}, restated for convenience.}

\simp*

\begin{proof} We have
\begin{align*}
[\forall \Phi \land \forall \Phi'] & = 
[\forall \Phi] \cap [\forall \Phi']
= \forall [\Phi] \cap \forall [\Phi']\\
&= \forall ([\Phi]\cap [\Phi'])
= \forall[\Phi\land \Phi']
= [\forall (\Phi\land \Phi')],
\end{align*}
where the third equality holds since, given an infinite trace $\tau$
and sets $T,T'$ of finite traces, 
all prefixes of $\tau$ are contained in $T$ and 
all prefixes of $\tau$ are contained in $T'$ iff
all prefixes of $\tau$ are contained in $T$ and in $T'$.
The proof of $\exists \Phi \lor \exists \Phi' \equiv  \exists (\Phi\lor \Phi')$
is analogous.

To see that the \LTLf formula $\Phi$ and the \LTLfp formula $\exists\Phi$ are equi-realizable,
let $\pi$ be a finite trace that satisfies $\Phi$. Then for any infinite extension $\tau$ of $\pi$,
we have that $\pi$ is a prefix of $\tau$ so that $\tau$ satisfies $\exists\Phi$.
For the converse direction, let $\tau$ be an infinite trace that satisfies $\exists\Phi$. Then
there is a prefix $\pi$ of $\tau$ that is contained in $[\Phi]$. Hence $\pi$ satisfies $\Phi$.
\end{proof}

\section*{Supplement to ``From Obligations to DWA''}

\paragraph{Full proof of Lemma~\ref{lem:weak_closure}, restated for convenience.}

\closure*

\begin{proof}
Let $\mathcal{A}_1=(T_1,F_1)$ and $\mathcal{A}_2=(T_2,F_2)$ be weak automata with transition systems 
$T_1=(\Sigma,Q_1,I_1,\delta_1)$ and $T_2=(\Sigma,Q_2,I_2,\delta_2)$.
We claim that:
\begin{itemize}
\item[--] $\mathcal{A}_{(1\land 2)} = (T_1\otimes T_2,F_1\times F_2)$ is s.t.\ $L(\mathcal{A}_{(1\land 2)})=L(\mathcal{A}_1)\cap L(\mathcal{A}_2)$;
\item[--] $\mathcal{A}_{(1\lor 2)} =(T_1\otimes T_2,F_1\times Q_2\cup Q_1\times F_2)$ is s.t.\ $L(\mathcal{A}_{(1\lor 2)}) =L(\mathcal{A}_1)\cup L(\mathcal{A}_2)$;
\item[--] $\mathcal{A}_{(\lnot 1)} = (T_1,Q_1\setminus F_1)$ is s.t.\ $L(\mathcal{A}_{(\lnot 1)})=\Sigma^\omega\setminus L(\mathcal{A}_1)$.
\item[--]$\mathcal{A}_{(1\land 2)}$, $\mathcal{A}_{(1\lor 2)}$, and $\mathcal{A}_{(\lnot 1)}$ are weak automata.
\end{itemize}
Let $w\in\Sigma^\omega$ be a word and let $\pi_1=q_1 q_2 \ldots$, $\pi_2=q'_1 q'_2\ldots$ and $\pi_1\times\pi_2=
(q_1,q'_1)(q_2,q'_2)\ldots$ be the runs
of $T_1$, $T_2$ and $T_1\otimes T_2$ on $w$, respectively.
 
 The first claim follows since
 $\mathsf{Inf}(\pi_1\times \pi_2)\cap F_1\times F_2\neq \emptyset$ if and only if
$\mathsf{Inf}(\pi_1)\cap F_1\neq \emptyset$ and $\mathsf{Inf}(\pi_2)\cap F_2\neq \emptyset$.
For the second claim, we point out that $\mathsf{Inf}(\pi_1\times \pi_2)\cap (F_1\times Q_2\cup Q_1\times F_2)\neq \emptyset$ if and only if
$\mathsf{Inf}(\pi_1)\cap F_1\neq \emptyset$ or $\mathsf{Inf}(\pi_2)\cap F_2\neq \emptyset$.
For the third claim, we have $\mathsf{Inf}(\pi_1)\cap (Q_1\setminus F_1)\neq \emptyset$ if and only if
$\mathsf{Inf}(\pi_1)\cap F_1=\emptyset$ by weakness of $\mathcal{A}_1$, that is, since every strongly connected component
is either fully accepting or fully rejecting so that either $\mathsf{Inf}(\pi_1)\subseteq F_1$ or
$\mathsf{Inf}(\pi_1)\subseteq Q_1\setminus F_1$.

It remains to show that $(T_1\otimes T_2,F_1\times F_2)$, $(T_1\otimes T_2,F_1\times Q_2\cup Q_1\times F_2)$ and $(T_1,Q_1\setminus F_1)$ are
weak automata. This is obvious for $(T_1,Q_1\setminus F_1)$. 

For the other two cases, we make the following observation about strongly
connected components in the product transition system $T_1\otimes T_2$. Let $(q_1,q_2)$ and $(q'_1,q'_2)$ be two states that belong to the same
strongly connected component in $T_1\otimes T_2$, that is, let there be a loop through $(q'_1,q'_2)$ and $(q_1,q_2)$. 
Then it follows from the definition of $T_1\otimes T_2$ that
there is a loop through $q_1$ and $q'_1$ in $T_1$ and a loop through $q_2$ and $q'_2$ in $T_2$.
In other words, $q_1$ and $q'_1$ belong to the same strongly connected component in $T_1$,
and $q_2$ and $q'_2$ belong the same strongly connected component in $T_2$.
By weakness of $T_1$ and $T_2$, we then have $q_1\in F_1$ iff $q'_1\in F_1$
and $q_2\in F_2$ iff $q'_2\in F_2$.

To see that $(T_1\otimes T_2,F_1\times F_2)$ is a weak automaton, consider two states $(q_1,q_2)$ and $(q'_1,q'_2)$ that belong to the
same strongly connected component in $T_1\otimes T_2$. From the above argumentation, we have $(q_1,q_2)\in F_1\times F_2$ iff
$q_1\in F_1$ and $q_2\in F_2$ iff $q'_1\in F_1$ and $q'_2\in F_2$ iff 
$(q'_1,q'_2)\in F_1\times F_2$, as required.

To see that $(T_1\otimes T_2,F_1\times Q_2\cup Q_1\times F_2)$ is a weak automaton, consider two states $(q_1,q_2)$ and $(q'_1,q'_2)$ that belong to the
same strongly connected component in $T_1\otimes T_2$. Again, we have $(q_1,q_2)\in F_1\times Q_2\cup Q_1\times F_2$ iff
$q_1\in F_1$ or $q_2\in F_2$ iff $q'_1\in F_1$ or $q'_2\in F_2$ iff 
$(q'_1,q'_2)\in F_1\times Q_2\cup Q_1\times F_2$, as required.
\end{proof}

\paragraph{Full proof of Lemma~\ref{lem:ftc_to_automaton}, restated for convenience.}

\ftctoaut*

\begin{proof}
\begin{itemize}

\item[--] $\mathbb{Q}=\exists$:
Let $\tau\in (2^{AP})^\omega$ be an infinite trace. Then
we have $\tau\in[\exists\Phi]$ iff there is some finite prefix of
$\tau$ that is contained in $[\Phi]$. This in turn is (by equivalence of $\Phi$ and $\mathcal{A}_\Phi$) the case
iff there is some finite prefix of
$\tau$ that is contained in $L(\mathcal{A}_\Phi)$, which is the case iff the run of $\mathcal{A}_{\mathbb{Q}\Phi}=(T_\Phi,F)$ on $\tau$ 
eventually visits a state from $F$ iff $\tau\in L(\mathcal{A}_{\mathbb{Q}\Phi})$.
The last equivalence holds since all accepting states in $T^+_\Phi$ are sinks so that
a run of $T^+_\Phi$ on an infinite word
eventually visits $F$ iff it eventually stays in $F$ forever.

\item[--] $\mathbb{Q}=\forall$: Let $\tau\in (2^{AP})^\omega$ be an infinite trace. Then
we have $\tau\in[\forall\Phi]$ iff all finite prefixes of
$\tau$ are contained in $[\Phi]$. This in turn is (by equivalence of $\Phi$ and $\mathcal{A}_\Phi$) the case
iff all finite prefixes of
$\tau$ are contained in $L(\mathcal{A}_\Phi)$, which is the case iff the run of $\mathcal{A}_{\mathbb{Q}\Phi}=(T_\Phi,F)$ on $\tau$ 
does not visit a state that is not contained in $F$ iff $\tau\in L(\mathcal{A}_{\mathbb{Q}\Phi})$.
The last equivalence holds since all rejecting states in $T^-_\Phi$ are sinks so that a run of $T^-_\Phi$ on an infinite word
does not visit a state that is not contained in $F$ iff it eventually stays in $F$ forever.
\end{itemize}
To see that $(T^+_\Phi,F)$ is a weak automaton,
observe that any strongly connected component in $T^+_\Phi$
either consists of a single accepting (sink) state or exclusively of rejecting states;
this is the case since in $T^+_\Phi$ it is not possible to reach a rejecting state from an accepting state.
The argument showing the weakness of $(T^-_\Phi, F)$ is dual.
\end{proof}

\section*{Supplement to ``Synthesis via DWA''}

\paragraph{Proof of Theorem~\ref{thm:synth}, restated for convenience.}

\synth*

\begin{proof}
By Lemmas~\ref{lem:weak_closure} and~\ref{lem:ftc_to_automaton}, the DWA $\mathcal{A}_\Psi$
is equivalent to $\Psi$. Thus the system player wins the game induced
by $\mathcal{A}_\Psi$ if and only if there is a (synthesis) strategy
$\sigma$ such that every outcome of $\sigma$ satisfies $\Psi$.
By Lemma~\ref{lem:weakSCCsolution}, winning regions and strategies in the induced 
game can be computed in symbolic time
$\mathcal{O}(|Q_\Psi|)$, where $|Q_\Psi|\in 2^{2^{\mathcal{O}(|\Psi|)}}$.
\end{proof}

\paragraph{Symbolic algorithms for B\"uchi and co-B\"uchi games.}

For completeness, we include pseudo-code for the solution algorithms
for B\"uchi games (Algorithm~\ref{alg:Buechi}) and co-B\"uchi games
(Algorithm~\ref{alg:co-Buechi}). The algorithms
compute nested least and greatest fixpoints.

\begin{algorithm}[bt]
\caption{\label{alg:Buechi}$
\textsc{SolveB\"uchi}(V,F)$}
\DontPrintSemicolon
    $X=F$;\,$X'=V$\;
   \While{$X\neq X'$}{
       $X=X'$\;
       $\mathsf{target}=F\cap \mathsf{Cpre}_s(X)$\;
       $Y=V$;\,$Y'=\emptyset$\;
       \While{$Y\neq Y'$}{
           $Y=Y'$\;
           $Y'=\mathsf{target}\cup \mathsf{Cpre}_s(Y)$\;
       }
       $X'=Y$\;
   }
   \Return{$X$};
\end{algorithm}

\begin{algorithm}[bt]
\caption{\label{alg:co-Buechi}$
\textsc{SolveCo-B\"uchi}(V,F)$}
\DontPrintSemicolon
    $X=V$;\,$X'=\emptyset$\;
   \While{$X\neq X'$}{
       $X=X'$\;
       $\mathsf{target}=\mathsf{Cpre}_s(X)$\;
       $Y=\emptyset$;\,$Y'=V$\;
       \While{$Y\neq Y'$}{
           $Y=Y'$\;
           $Y'=\mathsf{target}\cup (F\cap\mathsf{Cpre}_s(Y))$\;
       }
       $X'=Y$\;
   }
   \Return{$X$};
\end{algorithm}

\paragraph{Full proof of Lemma~\ref{lem:weakSCCsolution}, restated for convenience.}

\scc*

\begin{proof}
For one direction of the proof, we use the data computed by the algorithm to construct a memoryless system player strategy 
that wins all game nodes from the computed set $W$.
For each SCC $S$, the algorithm computes a memoryless strategy
$\sigma_S$ for the system player. Let $v\in W$ and let $S=SCC(v)$ denote the SCC of $v$. If $S$ is rejecting, then $\sigma_S$ ensures that every play starting at $v$
eventually leaves $S$ to a node $w\in W$ that belongs to a lower SCC.
If $S$ is accepting, then $\sigma_S$ ensures that every play starting at $v$
either stays within $S$ forever, or eventually leaves $S$ to a node $w\in W$ that belongs to a lower SCC.
We define the memoryless strategy $\sigma$ for the overall game to always play according to the strategy for the current SCC. Formally,
put $\sigma(v)=\sigma_{SCC(v)}(v)$ for each game node $v\in V_s$.
It follows that $\sigma$ wins every node in $W$.

For the converse direction, let $v$ be a game node that is won by the environment player and let $\sigma$ be an environment strategy
such that every play starting at $v$ and following $\sigma$ eventually stays forever within a rejecting SCC in $A$. We show that $v\notin W$.
The proof proceeds by induction on the number $l$ of different SCCs that are reachable from $v$. If $l=1$, then there is no
way to leave $SCC(v)$, and $SCC(v)$ is rejecting by assumption. Hence the algorithm treats $SCC(v)$ in the first
iteration of the loop (lines 2--10), using
lines 8--9. We have $v\notin \text{Reach}(SCC(v),\emptyset)$ so that $v\notin W$, as required.
If $l>1$, then we distinguish cases. If $SCC(v)$ is accepting, then the environment has, by assumption, a strategy that ensures
that every play starting at $v$ eventually leaves $SCC(v)$ by reaching a node $w$ belonging to a lower SCC.
Every such node $w$ is won by the environment and, by the inductive hypothesis, not contained in $W$.
The algorithm treats the accepting $SCC(v)$ using
lines 6-7. We have $v\notin \text{Safe}(SCC(v),W)$ so that $v\notin W$, as required.
If $SCC(v)$ is rejecting, then the environment has a strategy that ensures that every play starting at $v$ either stays in $SCC(v)$ forever
or eventually leaves the $SCC(v)$ to a node $w$ belonging to a lower SCC. 
Every such node $w$ is won by the environment and, by the inductive hypothesis, not contained in $W$.
Again, the algorithm treats $SCC(v)$ using
lines 8-9 and we have $v\notin \text{Reach}(SCC(v),W)$ so that $v\notin W$, as required.

Regarding time complexity, graphs with $n$ vertices can be decomposed into their SCCs in symbolic time 
$\mathcal{O}(n)$~\cite{DBLP:conf/tacas/LarsenSSJPP23}.
For each strongly connected component $SCC$, Algorithm~\ref{alg:weakSCC}
computes exactly one of the sets
$\text{Reach($SCC,W$)}$ or $\text{Safe($SCC,W$)}$, each of which
can be obtained in symbolic time $\mathcal{O}(|SCC|)$. Since the sizes of all SCCs 
sum to $n$, the overall runtime is linear.
A memoryless strategy is constructed
by always following the memoryless safety or reachability strategy for the current SCC. 

\end{proof}

\end{document}